\documentclass[a4paper]{article}
\usepackage{amsmath,amsfonts,amsthm}
\usepackage{fullpage}
\usepackage{xspace}
\usepackage{tikz}
\usepackage{enumitem}
\usepackage{multicol}
\usetikzlibrary{matrix, arrows,automata}
\newtheorem{theorem}{Theorem}

\usepackage{authblk}

\usepackage[utf8]{inputenc}
\usepackage[OT4]{fontenc}

\usepackage{graphicx}
\usepackage{xcolor}

\usepackage{amsfonts}
\usepackage{amssymb}
\usepackage{mathtools}

\usepackage{epsfig}

\usepackage[ruled,vlined,linesnumbered]{algorithm2e}

\usepackage{hyperref}

\definecolor{darkgreen}{rgb}{0,0.4,0}

\newtheorem{fact}{Fact}
\newtheorem{lem}{Lemma}

\newcommand{\ls}{\leqslant}

\newcommand{\SP}{\textsf{P}\xspace}
\newcommand{\Init}{\textsf{Initial}\xspace}
\newcommand{\Roam}{\textsf{Roam}\xspace}
\newcommand{\Rr}[1]{\textsf{RR}_{#1}}
\newcommand{\Down}{\textsf{Down}\xspace}
\newcommand{\Up}{\textsf{Up}\xspace}
\newcommand{\Terminated}{\textsf{Terminated}\xspace}
\newcommand{\port}{\textsf{last}\xspace}
\newcommand{\parent}{\textsf{parent}\xspace}
\newcommand{\roott}{\textsf{root}\xspace} 
\newcommand{\Root}{\textsf{Root}\xspace} 
\newcommand{\Act}{\textsf{Act}\xspace}
\newcommand{\Ag}{\textsf{Ag}\xspace}
\newcommand{\Tok}{\textsf{Tok}\xspace}
\newcommand{\NULL}{\textsf{NULL}\xspace}
\newcommand{\Path}{\textsf{Path}\xspace}
\newcommand{\visited}{\textsf{visited}\xspace}
\newcommand{\MOVE}[1]{\textsf{MOVE}{#1}\xspace}
\newcommand{\DROP}[1]{\textsf{DROP}{#1}\xspace}
\newcommand{\TAKE}[1]{\textsf{TAKE}{#1}\xspace}
\newcommand{\State}{\textsf{State}}

\newcommand{\etal}{{\it et~al.}\xspace}

\SetProcNameSty{textit}
\SetProcArgSty{textit}
\newcommand{\hrulealg}[0]{\vspace{1mm} \hrule width 180pt height .1pt \vspace{1mm}}

\newcommand{\customparagraph}[1]{\noindent\textsf{\textbf{#1}}\hspace*{3mm}}

\newcommand{\cleanmem}{\textsf{CleanMem}\xspace}
\newcommand{\dirtymem}{\textsf{DirtyMem}\xspace}
\newcommand{\token}{\textsf{Token}\xspace}

\title{Tree exploration in dual-memory model}

	\author[1]{Dominik Bojko}
	\author[1]{Karol Gotfryd}
	\author[2]{Dariusz R. Kowalski}
	\author[1]{Dominik Pajak}
	\affil[2]{School of Computer and Cyber Sciences, Augusta University, USA, dkowalski@augusta.edu}
	\affil[1]{Wrocław University of Science and Technology, Poland, dominik.bojko@pwr.edu.pl, karol.gotfryd@pwr.edu.pl, dominik.pajak@pwr.edu.pl}

\begin{document}
\maketitle
\begin{abstract}
We study the problem of online tree exploration by a deterministic mobile agent. Our main objective 
is to establish what features of the model of the mobile agent and the environment allow linear exploration time. We study agents that, upon entering to a node, do not receive as input the edge via which they entered. In such a model, deterministic memoryless exploration is infeasible, hence the agent needs to be allowed to use some memory. The memory can be located 
at the agent or at each node. The existing lower bounds show that if the memory is either only at the agent or only at the nodes, then the exploration needs superlinear time. We show that tree exploration in dual-memory model, with constant memory at the agent and logarithmic at each node is possible in linear time when one of two additional features is present: 
fixed initial state of the 
memory at each node (so called clean memory) or a single movable token. We present two algorithms working in linear time for arbitrary trees
in
these two models. 
On the other hand, in our lower bound we show that if the agent has a single bit of memory and one bit is present at each node, then exploration may require quadratic time on paths, 
if the initial memory at nodes could be set arbitrarily (so called dirty memory).
This shows that having clean node memory or a token allows linear exploration of trees in the model with two types of memory, but having neither of those features may lead 
to quadratic exploration time even on a simple path.
\end{abstract}

\section{Introduction}
\vspace*{-1mm}
Consider a mobile entity 
deployed inside an undirected graph with the objective to visit all its nodes and traverse all its edges, without any a priori knowledge of the topology of the graph or of its size. This problem, called a graph exploration, is among the basic problems investigated in the context of a mobile agent in a graph. In this paper we focus on tree exploration by a deterministic agent.  Clearly, to explore the whole tree 
of
$n$ nodes, any agent needs time~$\Omega(n)$. A question arises: ``What are the minimum agent capabilities for it to be able to complete the tree exploration in linear time?'' Many practical applications may not allow the agent to easily backtrack its moves, due to technical, security or privacy reasons --
hence, in this work we assume that the agent, upon entering a node, does \textbf{\em not} receive any information about the edge via which it entered. 

A deterministic agent at a node of degree $d$ needs at least $d$ different inputs in order to be able to choose any of $d$ outgoing edges. 
We consider three components of the input to the agent: memory at the agent, memory at each node, and movable tokens. 
In order to ensure the sufficient number of possible inputs, the number of possible values of the memory and possible present or absent states of the tokens, must exceed $d$. Hence, the sum of the memory at the agent, memory at the node and the total number of tokens must be at least $\lceil \log d_v\rceil$ bits, for any node of degree $d_v$ (we denote by $\log$ the base-$2$ logarithm). 

In many applications, mobile agents are meant to be small, simple and inexpensive devices or software. These restrictions limit the size of the memory with which they can be equipped. Thus it is crucial to analyze the performance of the agent equipped with the minimum size of the memory. In this paper we 
assume
the asymptotically minimum necessary memory size of $O(\log d_v)$, where only a constant number of bits is stored at the agent and logarithmic number at each node. Our result in the model with the tokens assumes the presence of only a single movable token.
%


The main focus of the paper is on deterministic tree exploration 
in the minimum possible time using the minimum possible memory. We consider two additional features of the model, namely -- clean memory or a token, and we show that each of these assumptions alone allows linear time tree exploration. We also prove 
that when both features are absent, then the linear time exploration with small memory
may be
impossible even on a simple path. 
\vspace*{-2mm}
\section{Model}
\vspace*{-1mm}
The agent is located in an initially unknown tree $T = (V,E)$ with $n = |V|$ nodes.  
The agent can traverse one of the edges incident to its current location within a single step, in order to visit a neighboring node at its other end.
The nodes of the tree are unlabelled, however in order for the agent to locally distinguish the edges outgoing from its current position, we assume that the tree is port-labelled. This means that at each node $v$ with some degree $d_v$, its outgoing edges are uniquely labelled with numbers from $\{1, 2 ,\dots,d_v\}$. Throughout the paper we assume that the port labels are assigned by an adversary that knows the algorithm used by the agent and wants to maximize the exploration time. 
\vspace*{1mm}\\
\customparagraph{Memory.}
The agent is endowed with some number of memory bits (called \emph{agent memory} or \emph{internal memory}), 
  which it can access and modify. In our results, we assume that the agent has $O(1)$ bits of memory, hence a constant number of states. 

Each node $v \in V$ contains some number of memory bits that can be modified by the agent when visiting that node. We will call these bits \emph{node memory} or \emph{local memory}. We assume that node $v$ contains $\Theta(\log d_v)$ 
bits of memory. 
Hence, each node can store a constant number of pointers to its neighbors.
\vspace*{1mm}\\
\customparagraph{Algorithm.}
Upon entering a node $v$, the agent $a$ receives, as input, its current state $S_a$, the state $S_v$ of the current node $v$, the number 
$T_v$ of tokens at $v$, the number $T_a$ of tokens at the agent and the degree $d_v$ of the current node. It outputs its new state $S'_a$, new state $S'_v$ of node $v$, new state $T'_a, T'_v$ of the tokens at the agent/node, respectively, and port number $p_{out}$ via which it exits node $v$; hence, the algorithm defines a transition:
\vspace*{-2mm}
\[
(S_a,S_v,T_a, T_v, d_v) \rightarrow (S'_a,S'_v, T'_a, T'_v,p_{out})
\]
In models without the tokens, the 
state transition
of the algorithm can be simplified to:
\vspace*{-2mm}
\[
(S_a,S_v, d_v) \rightarrow (S'_a,S'_v,p_{out})
\]
\customparagraph{Starting state.} 
We assume that for a given deterministic algorithm, there is one starting state $s$ of the agent (obtained from the agent memory), in which the agent is at the beginning of (any execution of) the algorithm. The agent can also use this state later (it can transition to $s$ in subsequent steps of the algorithm). 
The starting location of the agent is chosen by an adversary accordingly to agent's algorithm (to maximize the agent's exploration time).
\\
\customparagraph{Our models.}
In this paper we study the following three models:

In \cleanmem, there is a fixed state $\hat{s}$ and each node $v$ is initially 
in the state $S_v=\hat{s}$, regardless of the degree of the vertex. In this model the agent does not have access~to~any~tokens.

In \dirtymem, the memory at the nodes is in arbitrary initial states, 
which are chosen by an adversary. In this model the agent does not have access to any tokens.

In \token, the agent is initially equipped with a single token. If the agent holds a token, it can drop it at a node upon 
visiting
that node. 
When visiting a node, the agent receives as input whether the current node already contains a token.
In this case the agent additionally outputs whether it decides to pick up the token, i.e., deduct from $T_v$ and add to $T_a$.
Clearly, since the agent has only one token, then after dropping it at some node, it needs to pick it up before dropping it again at some other node.
In this model we assume that the initial state of the memory at each node is chosen by an adversary (like in \dirtymem model).

\section{Our results}
\customparagraph{Upper bounds.} 
While most of the existing literature focus on feasibility of the exploration, we show that it is possible to complete the tree exploration in the minimum possible linear time using (asymptotically) minimal memory.
We show two algorithms in models \cleanmem and \token, exploring arbitrary unknown trees in the optimal time $O(n)$ if constant memory is located at the agent and logarithmic memory is located at each node.
Our results show that in the context of tree exploration in dual-memory model, the assumption about clean memory (fixed initial state of node memory) can be ``traded'' for a single token.

It is worth noting that in both our algorithms, the agent returns to the starting position and terminates after completing the exploration, which is a harder task than a perpetual exploration.
If the memory is only at the agent, exploration of trees with stop at the starting node requires $\Omega(\log n)$ bits of memory~\cite{DiksFKP04}; if the agent is only required to terminate at any node then still a superconstant $\Omega(\log\log\log n)$ memory is required~\cite{DiksFKP04}; while exploration without stop is feasible using $O(\log \Delta)$ agent memory~\cite{DiksFKP04} (where $\Delta$ is the maximum degree of a node).
All these 
results assume that the agent could receive the port number via which it entered the current node, while our algorithms operate without this information.
\vspace*{2mm}\\
\customparagraph{Lower bound.} 
To explore a path with a single bit of memory and with arbitrary initial state at each node (note that $\log d_v=1$ for internal nodes on the path), one can employ the Rotor-Router algorithm and achieve exploration time of $O(n^2)$~\cite{YanovskiWB03} (there is no need for agent memory), which is time and memory optimal in the model with only node memory~\cite{MencPU17}.

In our lower bound we analyze line exploration with one bit at each node and additional one bit of memory at the agent.
We want to verify the hypothesis that dual-memory allows to achieve linear time of exploration.
We provide a partial argument for the contrary -- we prove that in \dirtymem model an exploration of the path requires $\Omega(n^2)$ steps in the worst case.
This shows that adding a single bit at the agent to the model with a single bit at each node does not reduce the exploration time significantly.

\subsection{Previous and related work}
\vspace*{-1mm}
\customparagraph{Only node memory.}
One approach for graph exploration that uses only node memory is Rotor-Router~\cite{YanovskiWB03}. It is a simple strategy, where upon successive visits to each node, the agent is traversing the outgoing edges in a round-robin fashion. Its exploration time is $\Theta(mD)$ for any graph with $m$ edges and diameter $D$~\cite{BampasGHIKKR17,YanovskiWB03}.
It is easy to see that this algorithm can be implemented in port-labelled graphs with zero bits of memory at the agent and only $\lceil \log d_v \rceil$ bits of memory at each node $v$ with degree $d_v$ (note that this is the minimum possible amount of memory for any correct exploration algorithm using only memory at the nodes). 
Allowing even unbounded memory at each node still leads to $\Omega(n^3)$ time for some graphs, and $\Omega(n^2)$ time for paths~\cite{MencPU17}. Interestingly, these lower bounds hold even if the initial state of the memory at each node is clean (i.e., each node starts in some fixed state $\hat{s}$ like in \cleanmem model).
Hence, having only memory at the nodes is insufficient for exploration of trees faster than in quadratic time. 
\vspace*{2mm}\\
\customparagraph{Only agent memory.}
When the agent is not allowed to interact with the environment, then such agent, when exploring regular graphs, does not acquire any new information during exploration. Hence such an algorithm, is practically a sequence of port numbers that can be defined prior to the exploration process. 
In the model with agent memory, the agent is endowed with some number of bits of memory that the agent can access and modify at any step.
If this number of bits is logarithmic in $n$ (otherwise,
the exploration is infeasible~\cite{DiksFKP04}), then in this model it is possible to implement Universal Traversal Sequences (the agent only remembers the position in the sequence). However, such sequences require $\Omega(n^{1.51})$ steps to explore paths~\cite{dai1996improved}.
Best known upper bound is $O(n^3)$~\cite{AleliunasKLLR79} and the best known constructive upper bound is $O(n^{4.03})$~\cite{Koucky03}.
Memory $\Theta(D \log \Delta)$ is sufficient and sometimes required to explore any graph with maximum degree $\Delta$~\cite{FraigniaudIPPP05}.
For directed graphs, memory $\Omega(n \log \Delta)$ at the agent is sometimes required to explore any graph with maximum outdegree $\Delta$, while memory $O(n\Delta\log \Delta)$ is always sufficient~\cite{FraigniaudI04}.
\vspace*{2mm}\\
\customparagraph{Finite state automata.}
An agent equipped with only constant number of bits of persistent memory (independent of the network size and topology and other parameters of the model) can be regarded as a finite state automaton
(see e.g. \cite{Cohen:2008:LabelGuidedFA,FraigniaudIPPP05,FraigniaudIRT05}). Movements of such agent, typically modeled as a finite Moore or Mealy automaton, are completely determined by a~state transition function 
$f(s,p,d_v) = (s',p')$, 
where $s$, $s'$ are the agent's
states and $p$, $p'$ are the ports through which the agent enters and leaves the node $v$. A finite state automaton cannot explore an arbitrary graph in the setting, where the nodes
have no unique labels \cite{Rollik:1980:AutomatenPlanaren}. Fraigniaud~\etal~\cite{FraigniaudIPPP05} showed that for any $\Delta \geqslant 3$ and any finite state agent with $k$ states, one can construct
a planar graph with maximum degree $\Delta$ and at most $k+1$ nodes, that cannot be explored by the agent. It is, however, possible in this model to explore (without stop) the trees, assuming that
the finite state agent has access to the incoming port number (cf. \cite{DiksFKP04}).
\vspace*{2mm}\\
\customparagraph{Tokens.} 
Deterministic, directed graph exploration in polynomial time using tokens has been considered in~\cite{BenderFRSV02}, where the authors show that a single token is sufficient if the agent has an upper bound on the total number of vertices and $\Theta(\log\log n)$ tokens 
are sufficient and necessary otherwise. In \cite{Bender:1994:PowerOfTeam} the authors proposed a probabilistic polynomial time algorithm that allows two cooperating agents without the knowledge of $n$ to explore any strongly connected
directed graph. They also proved that for a single agent with $O(1)$ tokens this is not possible in polynomial time in $n$ with high probability.    
In~\cite{FraigniaudIRT05} it was shown that using one pebble, the exploration with stop 
requires an agent with $\Omega(\log n)$ bits of memory. The same space bound remains true for perpetual exploration \cite{FraigniaudIPPP05}.
\vspace*{2mm}\\
\customparagraph{Two types of memory.}
Sudo~\etal~\cite{sudo} consider exploration of general graphs with two types of memory. They show that $O(\log n)$ bits at the agent and at each node allows exploration in time $O(m + nD)$.
Cohen~\etal~\cite{Cohen:2008:LabelGuidedFA} studied the problem of exploring an arbitrary graph by a finite state automaton, which is capable of assigning $O(1)$-bit labels to the nodes. They proposed
an algorithm, that -- assuming the agent knows the incoming port numbers -- dynamically labels the nodes with three different labels and explores (with stop) an arbitrary graph in time $O(mD)$ 
(if the graph can be labelled offline, both the preprocessing stage and the exploration take $O(m)$ steps). They also show that with $1$-bit labels and $O(\log \Delta)$ bits of agent memory it is 
possible to explore (with stop) all bounded-degree graphs of maximum degree $\Delta$ in time $O(\Delta^{10}m)$. 
\vspace*{2mm}\\
\customparagraph{Other approaches.}
If we assume that the port number via which the agent entered to the current node is part of the input to the algorithm (such a model is called a model with known inport) then 
tree exploration in linear time can be achieved by a simple algorithm known in the literature as \textit{basic walk}~\cite{GasieniecR08}.
This algorithm requires no memory (neither at the agent nor at the nodes).
However, if we require that the agent terminates at the starting point, then $\Omega(\log n)$ bits are required if there is only memory at the agent~\cite{DiksFKP04}.

Random walk is a classical and well-studied process, where the agent in each step moves to a neighbor chosen uniformly at random (or does not move with some constant probability). Exploration using this method takes expected time $\Omega(n \log n)$~\cite{Feige95a} and $O(n^3)$~\cite{Feige95}, where for each of these bounds there exists a graph class for which this bound is tight.
Randomness alone cannot ensure linear time of tree exploration, since expected time $\Omega(n^2)$ is required even for paths~\cite{lovasz1993random}. However approaches using memory at the agent~\cite{Kosowski13}, local information on explored neighbors~\cite{BerenbrinkCF15}, or local information on degrees~\cite{NonakaOSY10} have shown that there are many methods to speedup random walks.

Finally, to achieve fast tree exploration, usage of multiple parallel agents is possible. A~number of papers show that this approach allows to provide faster exploration than using only a single agent in the deterministic model with memory at the agent~\cite{DereniowskiDKPU15, FraigniaudGKP06, OrtolfS14}, with memory at the nodes~\cite{DereniowskiKPU16, KlasingKPS17,KosowskiP19} or using randomness~\cite{DereniowskiDKPU15, EfremenkoR09, ElsasserS11}.


\section{Upper bounds}
\label{sec:upper}
\subsection{Exploration in \cleanmem model}
\label{sec:cleanmem}
In this section we show an algorithm exploring any tree in $O(n)$ steps using $O(1)$ bits of memory at the agent and $O(\log d_v)$ bits of memory at each node $v$ of degree $d_v$ in \cleanmem model.

The memory at each node is organized as follows. It contains two port pointers (of $\lceil \log d_v \rceil$ bits each):
\vspace*{-1mm}
\begin{itemize}
\item $v.\parent$ -- this pointer (at some point of the execution) contains the port number leading to the root of the tree (i.e., the starting node of the exploration),
\item $v.\port$ -- this pointer points to the last port taken by the agent during the exploration
\end{itemize}
 and two flags (of $1$ bit):
 \begin{itemize}
 \item $v.\roott$ -- this flag indicates whether the node is $\Root$ of the exploration,
 \item $v.\visited$ -- this flag indicates whether the node has already been visited.
 \end{itemize}
 At each node, the initial state of $\port$ pointers is $1$, initial state of $\parent$ is $\NULL$ and initial state of each flag is $\texttt{False}$.
 
 The agent's memory contains one variable $\State$ that can take one of five possible values: \Init, \textsf{Exploring}, \textsf{ReturnFromParent}, \textsf{Completed}, $\Terminated$. 
 
 For simplicity of the pseudocode we use a flag $parentSet$. This flag does not need to be remembered because it is set and accessed in the same round. It is possible to write a more complicated pseudocode that does not require this variable.
 
 In our pseudocode, we use a procedure $\MOVE(p)$, in which the agent traverses an edge labelled with port $p$ from its current location. When the agent changes its state to $\Terminated$, it does not make any further moves.
\vspace*{2mm}\\
\customparagraph{Intuition of the algorithm.}
For the purpose of the analysis, assume that tree $T$ is rooted at the initial position of the agent. 
The main challenge in designing an exploration algorithm in this model is that the agent located at some node $v$ may not know which port leads to the parent of $v$. Indeed, if the agent knew which port leads to its parent, it could perform a DFS traversal. The agent would traverse the edges corresponding to outgoing ports in order $1,2,3,\dots,d_v$ (skipping the port leading to its parent) and take the edge to its parent after completing the exploration of the subtrees rooted at its current position. Note that it is possible to mark, which outgoing ports have been already traversed by the agent using only a single pointer at each node. 

 Since in our model, the agent does not know, which port leads to the parent, we need a second pointer $\parent$ at each node.
To set it correctly, we first observe that in the model \cleanmem, using flag $\visited$, it is possible to mark the nodes that have already been visited.
Notice that when the agent traverses some edge outgoing from $v$ in the tree for the first time and enters to a node that has already been visited, then this node is certainly the parent of $v$.
We can utilize this observation to establish correctly $\parent$ pointer at $v$ using state $\textsf{ReturnFromParent}$.
When entering to a node in this state, we know that the previously taken edge (port number of this edge is stored in $\port$) leads to the parent of $v$.
Our algorithm also ensures, that after entering a subtree, the agent leaves it once, and the next time it enters this subtree, it is entirely explored.
Finally, having correctly set pointers $\parent$ at all the nodes, allows the agent to efficiently return to the starting node after completing the exploration and flag $\roott$ allows the agent to terminate the algorithm at the starting node. 

\begin{algorithm}[t]
		\SetKw{KwAnd}{and}
		\SetKw{KwOr}{or}
		\caption{Tree exploration in \cleanmem model}
		\label{alg:cleanmem}
		\tcp{Agent is at some node $v$ and the outgoing ports are $\{1,2,\dots ,d_v\}$}
		\If{$\State  = \Init$}{
			$v.\roott \leftarrow \texttt{True},$
			$\State  \leftarrow \textsf{Exploring}, $
			 $v.\visited \leftarrow \texttt{True}$\;
			\MOVE{($v.\port$)}\;
		}
		\Else{
		
		\If{$\State = \textsf{ReturnFromParent}$}{
			$v.\parent \leftarrow \port$, $\State \leftarrow \textsf{Exploring}$\tcp*{mark edge as leading to parent}
			$parentSet \leftarrow \texttt{True}$\;}
			\lElse{
			$parentSet \leftarrow \texttt{False}$}
		\If{\label{line:bigif}(($\State  = \textsf{Exploring}$ \KwAnd $d_v = 1$) \KwOr $\State = \textsf{Completed}$ \KwOr $parentSet = \texttt{True}$) \KwAnd  $v.\roott = \texttt{False}$ \KwAnd $v.\port = d_v$ \KwAnd $v.\parent \neq \NULL$}{
			$\State  \leftarrow \textsf{Completed}$\label{line:completed}\tcp*{exploration of this subtree completed}
			\MOVE{$(v.\parent)$}\tcp*{return to the parent}\label{line:moveparent}}
		\ElseIf{$\State  = \textsf{Completed}$ \KwAnd $v.\roott = \texttt{True}$ \KwAnd $v.\port = d_v$}{
			$\State  \leftarrow \Terminated$\label{line:terminate}\tcp*{exploration of the whole tree completed}}
		\Else{
			\If{$\State = \textsf{Exploring}$ \KwAnd $v.\visited = \texttt{True}$ \KwAnd $parentSet = \texttt{False}$}{
				\label{line:setreturnfromparent}$\State \leftarrow \textsf{ReturnFromParent}$\;}
			\ElseIf{$\State = \textsf{Completed}$}{
				$\State \leftarrow \textsf{Exploring}$, $v.\port \leftarrow v.\port + 1$\label{line:incrementcompleted}\;
			}
			\Else{
			\lIf{ $v.\visited = \texttt{True}$}{
				$v.\port \leftarrow v.\port + 1$}
			\lElse{
				$v.\visited \leftarrow \texttt{True}$}
			}
		 	\MOVE{$(v.\port)$}\;\label{line:movelast}
			}
		}
\end{algorithm}

\begin{theorem}
\label{thm:clean}
Algorithm~\ref{alg:cleanmem} explores any tree and terminates at the starting node in $O(n)$ steps in \cleanmem model.
\end{theorem}
\begin{proof}
First note that the algorithm marks the starting node with flag $\roott$ (line 2). The algorithm never returns to state $\Init$, hence only this node will be marked with the $\roott$ flag.
The only line, where the agent terminates the algorithm is line \ref{line:terminate} hence the agent can only terminate in the starting node. We need to show that the agent will terminate in every tree and before the termination it will visit all the vertices and the time of the exploration will be $O(n)$.
Let us denote the starting node as $\Root$ and for any node $v$ different from $\Root$, will call the single neighbor of $v$ that is closer to $\Root$ as the parent of $v$. 

We will show the following claim using induction over the structure of the tree:

\noindent\textit{Claim 1: Assume that the agent enters to some previously unvisited subtree rooted at $v$ with $n_v$ vertices for the first time in state $\textsf{Exploring}$ in step $t_s$. Then the agent:
\begin{enumerate}[label={C.\arabic*}]
\item returns to the parent of $v$ for the first time in state $\textsf{Exploring}$ (denote by $t_r>t_s$ the step number of the first return from $v$ to its parent),\label{firstReturn}
\item goes back to $v$ in the state $\textsf{ReturnFromParent}$ at step $t_r +1$,\label{comeBack}
\item returns to the parent of $v$ for the second time in state $\textsf{Completed}$ (denote by $t_f>t_r +1$ the step number of the second return from $v$ to its parent),\label{secondReturn}
\item visits all the vertices of this subtree and spends $O(n_v)$ steps within time interval $[t_s,t_f]$.\label{timeSpend}
\end{enumerate}
}
We will first show it for all the leaves. Then, assuming that the claim holds for all the subtrees rooted at the children of some node $v$, we will show it for $v$. To show this claim for any leaf $l$, consider the agent entering in state $\textsf{Exploring}$ to $l$. Upon the first visit to $l$, the agent sets the flag $l.\visited$ to $\texttt{True}$ and leaves (without changing the state of the agent) with port $1$. This proves \ref{firstReturn}.
In the next step, at the parent of $l$, a state changes to $\textsf{ReturnFromParent}$ and the agent uses the same port as during the last time in parent of $l$.
Therefore the agent moves back to $l$ (\ref{comeBack}) and sets $l.\parent \gets 1$ (line 6). Then the agent changes its state to $\textsf{Completed}$ and moves to the parent (lines \ref{line:completed}-\ref{line:moveparent}).
This shows \ref{secondReturn}. Node $l$ was visited twice within the considered time steps, which shows \ref{timeSpend}. This completes the proof of the claim for all the leaves.

Now, consider any internal node $v$ of the tree (with $d_v >1$) and assume that the claim holds for all its children. When the agent enters to $v$ for the first time, it sets the flag $v.\visited$ to $\texttt{True}$ and moves to its neighbor $w$ (while being in state $\textsf{Exploring}$) via port $1$.


\textit{Case 1: $w$ is the parent of $v$.}
This immediately shows \ref{firstReturn}. If the parent of $v$ is not $\Root$, then it has the degree at least $2$, hence the agent will evaluate the if-statement in line \ref{line:bigif} to $\texttt{False}$. Thus the agent will execute line \ref{line:setreturnfromparent} and change its state to $\textsf{ReturnFromParent}$. The agent uses the same port as during the previous visit of $w$, therefore the agent goes back to $v$ (\ref{comeBack}), hence it will correctly set the pointer $v.\parent$ (it will point to the parent of $v$), which shows \ref{secondReturn}.

\textit{Case 2: $w$ is a child of $v$.}
In this case we enter to a child of $v$. We have by the inductive assumption (\ref{firstReturn}), that the agent will return from $w$ to $v$ in state $\textsf{Exploring}$.
Since the agent enters to $v$ in state $\textsf{Exploring}$, then the value of $parentSet$ is $\texttt{False}$ and the agent executes line~\ref{line:setreturnfromparent}, transitions to state $\textsf{ReturnFromParent}$ and then line~\ref{line:movelast} it takes the same edge as during the last visit to $v$, hence it moves back to node $w$ (\ref{comeBack}). By \ref{secondReturn} we get that the next time the agent will traverses the edge from $w$ to $v$, it will be in state $\textsf{Completed}$.
Then the agent increments the $v.\port$ pointer (line \ref{line:incrementcompleted}) and moves to the next neighbor of $v$ (line \ref{line:movelast}).

Thus the agent either finds the parent or explores the whole subtree rooted at one of its children $w$ in $O(n_w)$ steps (by the inductive assumption \ref{timeSpend}). An analogous analysis holds for ports $2,3,\dots,d_v$. When the agent returns from the neighbor connected to $v$ via the last port $d_v$, it is either in state $\textsf{Completed}$ or $\textsf{ReturnFromParent}$ (the second case happens if the edge leading from $v$ to its parent has port number $d_v$).
In both cases it executes lines \ref{line:completed} and \ref{line:moveparent} and leaves to its parent (the pointer to parent is established since the agent had traversed each outgoing edge, hence it moved to its parent and correctly set the $\textsf{parent}$ pointer).


By this way, the agent visits all the subtrees rooted at $v$'s children $w_1,w_2,\dots,w_{d_v-1}$ (or up to $w_{d_v}$ if $v=\Root$). Hence, the total number of steps for a node $v\neq \Root$ is $O\left(\sum\limits_{i=1}^{d_v-1} n_{w_i}\right) = O(n_v)$ and $v = \Root$ it is $O\left(\sum\limits_{i=1}^{d_v} n_{w_i}\right) = O(n)$.

To complete the proof of Theorem~\ref{thm:clean} we need to analyze the actions of the agent at $\Root$. Consider the actions of the agent at $\Root$ when the $\Root.\port$ pointer takes values $i = 1,2,\dots, d_{\Root}$. Let $v_i$ be the neighbor of $\Root$ pointed by port number $i$ at $\Root$.
Observe that the agent at $\Root$ behaves similarly as in all the other internal nodes (only exception is that the agent will never enter $\Root$ in state $\textsf{ReturnFromParent}$, because by Claim $1$ the agent can enter to it only in states $\textsf{Exploring}$ or $\textsf{Completed}$, since the $\Root$ is a parent of its every child).
By Claim 1, the agent visits the whole subtrees rooted at these nodes in time proportional to the number of nodes in these subtrees. When the agent returns from node $v_{d_{\Root}}$ in state $\textsf{Completed}$ then the agent terminates the algorithm in line \ref{line:terminate}. The total runtime is proportional to the total number of nodes in the tree.
\end{proof}

\subsection{Exploration in \token model}
\label{sec:token}
\customparagraph{Intuition of the algorithm.}
The agent has a single token, which can be $\DROP{ped}$, $\TAKE{n}$ and $\MOVE{d}$ (i.e., carried by the agent across an edge of the graph). Moreover, the agent has constant memory which contains a State from set $\{\Init, \Roam, \Rr{}, \Down, \Up, \Terminated\}$. Our algorithm ensures, that in states $\Roam$ and $\Rr{}$ the agent does not hold the token and in the remaining states, it holds the token. In these four states the agent can eventually $\DROP{}$ it. 
Each node $v$ has degree denoted by $d_v$, which is part of the input to the agent, when entering to a node. Moreover, the memory at each node $v$ is organized into three variables: $v.\port$ and $v.\parent$ of size $\lceil\log d_v\rceil$ and a flag $v.\roott\in\{\texttt{True},\texttt{False}\}$ to mark the starting node (the node with this flag set to $\texttt{True}$ will be called $\Root$). We assume that in all vertices, variables $\port$, $\parent$ and $\roott$ have initially any admissible value (if not, then the agent would easily notice it and change such a value).
Moreover, when the agent is entering to a node, it can see whether the node contains the token or not.

We would like to perform a similar exploration as in \cleanmem model -- we will use pointers $\port$ and $\parent$ as in Algorithm~\ref{alg:cleanmem}. However, the difficulty in this model is that these pointers may have arbitrary initial values. Especially, if the initial value of $\parent$ is incorrect, 
Algorithm~\ref{alg:cleanmem} may fall into an infinite loop. Hence in this section we propose a new algorithm that handles dirty memory using a single token. In this algorithm the agent maintains an invariant that the node with the token, and all the nodes on the path from the token's location to the root are guaranteed to have correctly set pointers to their parents. To explore new nodes, the agent performs a Rotor-Router traversal, starting from node $v$ with the token to its neighbor $w$ (by the invariant, the agent chooses $w$ as one of its children, not its parent). During this traversal, the agent is resetting the \parent pointers at each node (and cleaning the \roott flag). The agent is using \port as the pointer for the purpose of Rotor-Router algorithm. The agent does \textbf{not} have to reset \port as the Rotor-Router requires no special initialization. Moreover, by the properties of Rotor-Router the agent does not traverse the same edge twice in the same direction before returning to the starting node. Since the agent starts in the node with the token, it can notice that it completed a traversal. During this traversal each edge is traversed at most twice (once in each direction). Moreover, each node visited during this traversal has cleaned memory (pointer $\parent$ points to NULL and \roott is set to $\texttt{False}$). After returning to the node with the token, pointer $v.\port$ points at $w$ and $w.\port$ points at $v$. Hence it is possible to traverse this edge (in a special state \Down similar to \textsf{ReturnFromParent} from Algorithm~\ref{alg:cleanmem}) and correctly set pointer $w.\parent$ maintaining the invariant. After moving the token down, the agent starts the Rotor-Router procedure again. Since the agent is not cleaning the pointer \port, the Rotor-Router will use different edges than during the previous traversal. Thus, our algorithm traverses every edge at most $6$ times. 
\customparagraph{Preliminaries.}
 For brevity, let us introduce a variable $\token$ associated with a vertex currently occupied by the agent, which is $1$, when the agent meets the token and $0$ otherwise.
If $\token=1$, then the agent can $\TAKE{}$. If the agent holds the token, it can $\DROP{}$ it.

Additionally, we consider two substates of $\Rr{}$, depending on the value of $\token$ variable. Substate $\Rr{0}$ is state $\Rr{}$, if variable $\token = 0$ at the node to which the agent entered in the considered step. Similarly $\Rr{1}$ indicates that the agent enters in state $\Rr{}$ to a node with $\token = 1$. Note that this distinction is only for the purpose of the analysis, and it does not influence the definition of the algorithm. In our analysis, we will call $\Init, \Roam, \Rr{0}, \Rr{1}, \Down, \Up, \Terminated$ \emph{actions} as these states (and substates) correspond to different lines executed by the agent (see the pseudocode of Algorithm~\ref{alg:token}).
\begin{algorithm}[h]
\begin{multicols}{2}
		\SetKw{KwAnd}{and}
		\SetKw{KwOr}{or}
		\caption{Tree exploration in \token model}\label{alg:token}
		\tcp{Agent is at some node $v$ and the outgoing ports are $\{1,2,\dots ,d_v\}$}
		\If{$\State=\Init$}{
			$\textit{Clean()}$, $v.\port \leftarrow 1$, $\DROP{}$\; 
			\tcp{mark the root for termination}
			$v.\roott\leftarrow \texttt{True}$\;
			$\State\leftarrow \Roam$\;
		}
		\ElseIf{$\State=\Rr{}$ \KwAnd $\token=1$} 
		{
			\tcp{substate $\Rr{1}$}
			$\TAKE{}$, $\State \leftarrow \Down$\;
		}
		\ElseIf{$\State=\Rr{}$ \KwAnd $\token=0$} 
		{
			\tcp{substate $\Rr{0}$}
			$\textit{Clean()}$\;
			$\textit{Progress()}$ \tcp*{increment $\port$ \hspace{3mm}}
		}
		\ElseIf{$\State=\Down$}{
			$\DROP{}$, $v.\parent \leftarrow v.\port$\;
			$\textit{Progress()}$\;
			$\State \leftarrow \Roam$\;
			$\textit{IfUp()}$	\tcp*{check if in a leaf\hspace{3mm}}
		}
		\ElseIf{$\State=\Up$}{
			$\DROP{}$, $\textit{Progress()}$\;
			$\State \leftarrow \Roam$\;
			\If{$v.\port=1$ \KwAnd $v.\roott = \texttt{True}$}{
				$\State \leftarrow \Terminated$\;}
			$\textit{IfUp()}$\;
		}
		\ElseIf{$\State=\Roam$}{
			$\textit{Clean()}$, $\State \leftarrow \Rr{}$\;
		}
		\If{$\State\neq \Terminated$}{
			$\MOVE{(v.\port)}$\;}
		  \setcounter{AlgoLine}{0}
  \hrulealg
  \SetKwProg{myproc}{Procedure}{}{}
  \myproc{Clean()}{
			$v.\roott \leftarrow \texttt{False}$, $v.\parent \leftarrow \NULL$\;}
  \setcounter{AlgoLine}{0}
  \hrulealg
  \SetKwProg{myproc}{Procedure}{}{}
  \myproc{Progress()}{
		$v.\port\leftarrow (v.\port \mod\; d_v) +1$\;}
  \setcounter{AlgoLine}{0}
  \hrulealg
  \SetKwProg{myproc}{Procedure}{}{}
  \myproc{IfUp()}{
			\If{$v.\parent=v.\port$}{
				$\TAKE{}$, $\State\leftarrow\Up$\;
			}}
\vspace*{0.1mm}
\end{multicols}
\end{algorithm}

We assume that $\Init$ action is performed at time step $0$.
Let $\Root$ denote the initial position of the agent. Let $v.\port(t)$ and $v.\parent(t)$ denote, respectively, the values of $v.\port$ and $v.\parent$ pointers at the end of the moment $t$. If $v$ is the starting point of the $t$-th step, then we say that $v.\port(t)$ is the outport related to $t$-th moment. Moreover, $v.\port(\cdot)$ cannot be changed until the node $v$ will be visited for the next time.
Each $\MOVE{(\port)}$ involves traversing an edge outgoing from the current position of the agent via port indicated by the current value of variable $\port$ at the current position.
Let $\Path(v)$ denote a set of all vertices, which are on the shortest path from $\Root$ to $v$, excluding $\Root$. Let $T_v$ denote the subtree of $T$ rooted at $v$, i.e., $(w\in T_v)\, \equiv\, (w=v\, \lor\, v\in \Path(w))$. Let $T_{v,p}$ denote a subtree of $T_v$, rooted at a node connected by edge labelled by outport $p$ at vertex $v$, i.e. if $p$ leads from $v$ to $w$ (where $v\in \Path(w)$), then $T_{v,p}=T_{w}$ (we do not define such the tree when $p$ directs towards $\Root$). We say that the action is performed away from $\Root$ (downwards), if it starts in $v$ and ends in $T_v$. Otherwise the action is towards $\Root$ (upwards). Let $\Tok(t)$ and $\Ag(t)$ be the positions of the token and the agent, respectively, at the end of the moment $t$. Let $\Act(t)$ denote the action performed during the $t$-th step. Let the variable $v.\visited(t) \in \{0,1\}$ indicates whether $v$ was visited by Algorithm \ref{alg:token} until the moment $t$. Note that this variable is not stored at the nodes and this notation is only for the analysis. 
\vspace*{1mm}

\noindent
\customparagraph{Properties of the algorithm.}
Let us define the following set of properties $\SP(t)$ (in this definition we denote $\Ag(t)=w$) that describe the structure of the walk and the interactions with the memory at the nodes by an agent performing Algorithm~\ref{alg:token}. 
\begin{enumerate}[leftmargin=2.5em, label={P.\arabic*}]
\item During $t$-th step \begin{enumerate}
\item $\Down$ and $\Roam$ actions are performed away from $\Root$ (downwards).
\item $\Up$ and $\Rr{1}$ actions are performed towards $\Root$ (upwards).
\end{enumerate} \label{RootDir}
\item $w \in T_{\Tok(t)}$. \label{AgUnderTok}
\item If $w$ is visited for the first time at step $t$ ($\Act(t)\in\{\Init, \Roam, \Rr{0}\}$), then it is also cleaned, which means that $w.\parent$ is set to $\NULL$ and $w.\roott$ is set to $\texttt{False}$.\label{FirstClean}
\item For every $v\in \Path(\Tok(t))$, $v.\parent(t)\neq \NULL$. \label{ParOverTok}
\item If $w.\visited(t-1)=1$ and $w.\parent(t-1)\neq \NULL$, then $\Act(t)\not\in\{\Roam, \Rr{0}\}$. \label{CleanOnce}
\item If $w.\visited(w,t)=1$ and $w.\parent(t)\neq \NULL$, then $w.\parent(t)$ directs towards $\Root$. \label{OnePar}
\item 
If the state at the end of $t$-th moment is $\Up$, then $T_{w}$ is explored.\label{EndUp}
\item If $t>0$, then for every $x\ls \Root.\port(t-1)$, $T_{\Root,x}$ is explored before $t$-th moment.
\label{RootProg}
\item If $w.\parent(t)= \NULL$, then there do not exist $s_1<s_2< t$ such that $w.\port(s_1)=w.\port(t)\neq w.\port(s_2)$, where $w.\parent(s_1)= \NULL$ ($w.\port(\cdot)$ cannot be changed to the same value two times before $w.\parent(\cdot)$ was established). \label{PortOnceBefore}
\item There do not exist $s_1<s_2< t$ such that $w.\port(s_1)=w.\port(t)\neq w.\port(s_2)$ and $w.\parent(s_1)\neq \NULL$ ($w.\port(\cdot)$ cannot be changed to the same value two times after $w.\parent(\cdot)$ was established). \label{PortOnceAfter}
\end{enumerate}%
\begin{lem}
\label{lem:sp0}
Right before the $1-st$ moment, $\Init$ phase of Algorithm \ref{alg:token} guarantees $\SP(0)$.
\end{lem}
\begin{proof}
\ref{RootDir}$(0)$ does not depend on $\Init$. $\Root.\port(0)=1$ entails \ref{RootProg}$(0)$, because no other nodes were visited. \ref{AgUnderTok}$(0)$, \ref{FirstClean}$(0)$, \ref{EndUp}$(0)$ and \ref{PortOnceBefore}$(0)$ are trivially true. Note that $\Root\not\in \Path(\Tok(0))$, since $\Root$ is excluded from any such a path, hence \ref{ParOverTok}$(0)$. Moreover, $\Root.\parent(0)=\NULL$, so if-conditions of \ref{PortOnceAfter}$(0)$, \ref{CleanOnce}$(0)$ and \ref{OnePar}$(0)$ are not fulfilled, so they are true.
\end{proof}

\begin{lem}
\label{lem:spt}
If Algorithm \ref{alg:token} satisfies $\SP(s)$ for all $s\ls t$, then it also fulfills $\SP(t+1)$.
\end{lem}
\begin{proof}
In all the cases below we assume that the agent starts from node $v$ at the end of $t$-th moment and during $(t+1)$-th step $\MOVE{s}$ through $v.\port(t)$ to node $w$ and performs an action, depending on the case.

CASE I. $\Act(t+1)=\Roam$\\
If $\Act(t+1)=\Roam$, then $\Act(t)\in\{\Down, \Up, \Init\}$ and consequently $\Act(t)$ $\DROP{s}$ the token at $v$%
. If $v=\Root$, then \ref{RootDir}$(t+1)$ is true because the $\MOVE{}$ during $(t+1)$-st time step is clearly away from the root. Otherwise, from \ref{ParOverTok}$(t)$, we get $v.\parent(t)\neq \NULL$ and $\Act(t)\neq\Init$. Thence $\Act(t) \in \{\Down, \Up\}$ thus it end with IfUp() procedure, which would give $\Act(t+1)=\Up$, if we had $v.\port(t)= v.\parent(t)$, but since $\Act(t+1) = \Roam$ we get $v.\port(t)= v.\parent(t)$. Since due to \ref{OnePar}$(t)$, parent at $v$ is set correctly, hence $w\not\in\Path(v)$ and the $\MOVE{}$ during $(t+1)$-st time step is away from the root which gives us \ref{RootDir}$(t+1)$. 
Moreover $\Tok(t)=v$, so also \ref{AgUnderTok}$(t+1)$ is clearly true.
Notice that $\Roam$ always cleans $w$, and it neither moves the token nor changes $v.\port(\cdot)$, so \ref{FirstClean}$(t+1)$, \ref{ParOverTok}$(t+1)$, \ref{PortOnceBefore}$(t+1)$ and \ref{PortOnceAfter}$(t+1)$ are also true.
We already know that the $\MOVE{}$ during $(t+1)$-st time step was away from the $\Root$ thus \ref{RootProg}$(t+1)$ holds.\\
Assume, that $v\neq \Root$. If $w.\visited(t)=1$ and $w.\parent(t)\neq \NULL$, then there exists $s\ls t$ such that $w$ was visited by $\Act(s)=\Down$ action which is the only action which establishes $\parent$ pointer at $w$ (recall that upon the first visit at $w$ the pointer was cleaned), hence $\Tok(s)=w$. By \ref{ParOverTok}$(s)$ and \ref{RootDir}$(t+1)$, we get $v.\parent(s)\neq \NULL$. Moreover, by \ref{OnePar}$(t)$ and \ref{CleanOnce}$(t)$, the parent's pointer cannot be changed during $\Act(t+1)$. From \ref{RootDir}$(s)$ and \ref{RootDir}$(t)$, we have $v.\port(s)=v.\port(t)$, so by \ref{PortOnceAfter}$(t)$, $v.\port(\cdot)$ was not changed between the moments $s$ and $t$. However $v=\Tok(t)$, so by \ref{RootDir}$(s')$, for $s'\ls t$, there exists a moment $u$, such that $s<u<t$, $\Ag(u-1)=w$ and $\Act(u)=\Up$ (the only action which moves the token upwards).  Notice that $\Act(u)=\Up$ invokes Progress() method, which increases the port's pointer by $(1\mod\;d_v)$, so $v.\port(\cdot)$ did not change between the moments $s$ and $t$ and $v$ has to be a leaf of the tree. However, from \ref{RootDir}$(t+1)$, $w$ is a child of $v$, so it is not a leaf.
To complete the contradiction, showing \ref{CleanOnce}$(t+1)$, let us consider the opposite case, when $v=\Root$. Then, the proof is very similar. The only difference is that we need to use \ref{PortOnceBefore}$(t)$ instead of \ref{PortOnceAfter}$(t)$, because $\Root.\parent(\cdot)\equiv\NULL$, as $\Root.\parent(\cdot)$ can be established only by $\Down$ action, which cannot be performed towards $\Root$ from \ref{RootDir}$(s')$, for $s'\ls t$).\\
From \ref{CleanOnce}$(t+1)$ we know that $w.\parent(t)=\NULL$, so we simply get \ref{OnePar}$(t+1)$. \ref{EndUp}$(t+1)$ is independent of $\Roam$ action.

CASE II. $\Act(t+1)=\Rr{0}$\\
Properties \ref{RootDir}$(t+1)$ and \ref{EndUp}$(t+1)$ are independent of $\Rr{0}$ action.
Realize that $\Rr{0}$ action cannot be taken from $v$ to $\Root$, because of \ref{AgUnderTok}$(t)$, so $\Ag(t+1)\neq\Root$ and \ref{RootProg}$(t+1)$ is true.
If $\Act(t+1)=\Rr{0}$, then $\Act(t)\in\{\Rr{0}, \Roam\}$. Hence $v\neq \Tok(t)$ and \ref{AgUnderTok}$(t+1)$ remains fulfilled. $\Tok(t+1)=\Tok(t)$, so $w\not\in\Path(\Tok(t+1))$ and consequently \ref{ParOverTok}$(t+1)$ remains true. $\Rr{0}$ cleans $w$, so \ref{FirstClean}$(t+1)$ and \ref{OnePar}$(t+1)$ are obvious.\\
Note that $\parent(w,t+1)= \NULL$, so by \ref{CleanOnce}$(s')$, for $s'\ls t$, it was always $\NULL$ since the first visit of $w$.
This immediately shows \ref{PortOnceAfter}$(t+1)$ and \ref{CleanOnce}$(t+1)$.\\
Assume that $s_1$ is the first moment, when $w.\port(s_1)=w.\port(t+1)$ and $w.\port(\cdot)$ was changed meanwhile.
Since $w.\port(\cdot)$ can be incremented only by $(1\mod\;d_w)$, there exists a moment $s_2$ such that $s_1<s_2<t+1$ and $w.\port(s_2)$ directs upwards to some vertex $v_p$, but $w.\port(s_2-1)$ is not. Moreover, assume that $w$ was visited for the first time at moment $s_0\ls s_1$. This means that $\Ag(s_0-1)=v_p$.
By \ref{FirstClean}$(s_0)$, $\Act(s_0)\in\{\Roam, \Rr{0}\}$.
If $\Act(s_0)=\Roam$, then $v_p=\Tok(s_2)$, $\Act(s_2+1)=\Rr{1}$ and $\Act(s_2+2)=\Down$, so $w.\parent(s_2+2)\neq \NULL$, where $s_2+2\ls t+1$.
If $\Act(s_0)=\Rr{0}$, then there exists the first moment $s_3$, where $s_2<s_3\ls t+1$, such that $\Act(s_3)$ starts in $v_p$ and ends in $w$. 
Since $\Act(s_0)=\Rr{0}$ is preceded either by $\Roam$ or another $\Rr{0}$ action, therefore $v_p.\parent(s_0)=\NULL$. Moreover, we know that $v_p.\port(s_0-1)$ directs to $w$, so $v_p.\port(s_0-1)=v_p.\port(s_3)$. Note that from \ref{AgUnderTok}$(s')$ for $s_0\ls s'\ls s_2$, we obtain that the agent does not meet the token between the moments $s_0$ and $s_2$, so $\Act(s')=\Rr{0}$ for $s_0\ls s'\ls s_2+1$.
Therefore $v_p.\port(\cdot)$ is incremented by $(1\mod\;d_{v_p})$ at the moment $s_2+1$. Since $w$ is further from $\Root$ than $v_p$, $d_{v_p}>1$, and in consequence, as we know that $v_p.\port(s_0-1)=v_p.\port(s_3)$, then from \ref{PortOnceBefore}$(s')$, for $s'\ls t$, we conclude that $v_p.\parent(s_3)\neq \NULL$. Hence the $s_3$-th step is the $\Roam$ action and $w.\port(s_3)$ direct upwards. Therefore $\Act(s_3+1)=\Rr{1}$ and $\Act(s_3+2)=\Down$, and in consequence, $w.\parent(s_3+2)\neq \NULL$. This completes the proof of \ref{PortOnceBefore}$(t+1)$.

CASE III. $\Act(t+1)=\Rr{1}$\\
Realize that $w=\Tok(t)$, so $w.\visited(t)=1$ and \ref{FirstClean}$(t+1)$ follows from \ref{FirstClean}$(t)$. Moreover, from \ref{AgUnderTok}$(t)$, the agent has to move upwards during $\Act(t+1)$, so \ref{RootDir}$(t+1)$, what entails also \ref{AgUnderTok}$(t+1)$.
$\Rr{1}$ does not change any port or parent and does not move the token, so it is independent of \ref{RootProg}$(t+1)$, \ref{ParOverTok}$(t+1)$, \ref{CleanOnce}$(t+1)$, \ref{EndUp}$(t+1)$ and \ref{OnePar}$(t+1)$. Moreover then \ref{PortOnceAfter}$(t+1)$ and \ref{PortOnceBefore}$(t+1)$ follows from \ref{PortOnceAfter}$(t)$ and \ref{PortOnceBefore}$(t)$ respectively.

CASE IV. $\Act(t+1)=\Down$\\
Realize that $\Act(t)=\Rr{1}$ was performed upwards (due to \ref{RootDir}$(t)$) and $v.\port(t-1)$ was not changed since the last visit. Therefore $w=\Ag(t+1)=\Ag(t-1)$ and hence $(t+1)$-st step is performed downwards, what shows \ref{RootDir}$(t+1)$. Moreover, it entails that $\Root.\port(t)=\Root.\port(t+1)$, so $\Down$ action is independent of \ref{RootProg}$(t+1)$.
$\Down$ action $\DROP{s}$ the token, so \ref{AgUnderTok}$(t+1)$ holds.
Since $\Ag(t+1)=\Ag(t-1)$, \ref{FirstClean}$(t+1)$ follows from \ref{FirstClean}$(t)$.
$\Down$ action $\MOVE{s}$ the token downwards (from \ref{RootDir}$(t+1)$), but establishes $w.\parent(t+1)\neq \NULL$, so \ref{ParOverTok}$(t)$ gives \ref{ParOverTok}$(t+1)$.
Since $\Ag(t+1)=\Ag(t-1)$, $\Act(t)=\Rr{1}$ and $\Act(t-1)\in\{\Roam, \Rr{0}\}$, so $w.\parent(t-1)=w.\parent(t)=\NULL$. Therefore \ref{PortOnceAfter}$(t+1)$ is fulfilled. This also shows that \ref{PortOnceBefore}$(t+1)$ and \ref{CleanOnce}$(t+1)$ are independent of $\Down$ action.
$\Down$ action establishes $w.\parent(t+1)=w.\port(t-1)$, but $\Act(t)=\Rr{1}$, so from \ref{RootDir}$(t)$, $w.\parent(t+1)$ directs towards $\Root$. If the agent during this step changes its state to $\Up$, then $w.\parent(t+1)=(w.\parent(t+1) \mod\;d_w) + 1$, so $w$ is a leaf, so whole $T_w$ is explored and \ref{EndUp}$(t+1)$ is fulfilled.

CASE V. $\Act(t+1)=\Up$\\
Realize that $\Up$ state can be obtained only during the execution of IfUp() method, only when $v.\parent(t)=v.\port(t)$. From \ref{FirstClean}$(s')$, for $s'\ls t$, every parent's pointer of vertices visited until the $t$-th moment was cleaned, so \ref{OnePar}$(t)$ shows that $v.\port(t)$ directs upwards, which proves \ref{RootDir}$(t+1)$.\\
If $w\neq \Root$, then \ref{RootProg}$(t+1)$ does not depend on $\Up$ move.
Otherwise, let $w = \Root$, then $\Root.\port(\cdot)$ is increased at the moment $t+1$, so \ref{RootProg}$(t+1)$ follows from \ref{RootProg}$(t)$ if $\Root.\port(t)\neq d_{\Root}$. Otherwise, $\Root.\port(t+1)=1=\roott$, so Algorithm \ref{alg:token} is $\Terminated$.\\
Since $\Up$ is performed towards $\Root$ during the $(t+1)$-st step, $\Tok(t)=\Ag(t)\neq \Root$ had to be in $T_{\Root,\port(\Root,t)}$. If $\Act(t+1)$ ends in $\Root$, then \ref{EndUp}(t) shows that $T_{\Root,\Root.\port(t)}$ is explored, and from \ref{RootProg}$(t)$, we get \ref{RootProg}$(t+1)$.\\
As $\Up$ action $\DROP{s}$ the token, \ref{AgUnderTok}$(t+1)$ is obviously true, what also implies \ref{PortOnceBefore}$(t+1)$. $\Act(t+1)$ is performed upwards (by \ref{RootDir}$(t+1)$), so $w.\visited(t)=1$ and \ref{FirstClean}$(t+1)$ follows from \ref{FirstClean}$(s')$, for $s'\ls t$. $\Up$ action does not change parent's pointer, so it is independent of \ref{ParOverTok}$(t+1)$, \ref{OnePar}$(t+1)$ and also \ref{CleanOnce}$(t+1)$.\\
Assume that $s$ is the first moment when $w.\parent(s)\neq \NULL$. Therefore $\Act(s)=\Down$, so $w.\port(s)=(w.\parent(s) \mod\; d(w)) + 1$. Since the $w.\port(\cdot)$ can be only incremented by $1$ modulo $d(w)$, then the first repetition of the port's pointer may be only $w.\port(t+1)=w.\port(s)$ after $\Act(t+1)=\Up$. Assume that the node $w$ was visited last time at some step $s_2$. Then $w.\port(s_2)=w.\parent(s_2)$, which cannot change, according to \ref{OnePar}$(s')$, for $s_2\ls s'\ls t+1$, so in order to get back to $w$, $\Act(t+1)$ has to be preformed downwards, contrary to \ref{RootDir}$(t+1)$. This proves \ref{PortOnceAfter}$(t+1)$.\\
If $\Act(t+1)=\Up$ ends up in the state $\Up$, then $w.\parent(t) = (w.\port(t) \mod\; d(w)) + 1$ and $T_{v}=T_{w,w.\port(t)}$ is explored by \ref{EndUp}$(t)$. Moreover, since $w.\parent(t)\neq\NULL$, there exists the moment $s$, such that $v$ was visited during $\Act(s)=\Down$, when $w.\parent(\cdot)$ was established, so by \ref{PortOnceAfter}$(s)$ for $s\ls s'\ls t+1$, every outport of $w$ was used once (because $w.\port(s)=(w.\parent(s) \mod\;d_w)+1$), so $w$ was visited only upwards in steps from $s+1$ to $t$. From \ref{RootDir}$(s')$ and \ref{OnePar}$(s')$ for $s\ls s'\ls t+1$, $w.\port(\cdot)$ was changed only during $\Up$ actions, what implies that every $T_{w,w.\port(s')}$ was explored by \ref{EndUp}$(s')$ for $s\ls s'\ls t+1$, so $T_v$ was explored as well, what shows \ref{EndUp}$(t+1)$.
\end{proof}

We will use properties $\SP(t)$ to show correctness and time complexity of our Algorithm~\ref{alg:token}.%
\begin{figure}[ht!]
\centering
\begin{tikzpicture}[->]
 \node [draw, rectangle, text width=2.7em, text height=0.9em, label={[label distance=0.2em,align=center]180:$\textit{Clean()}$\\$\port=1$\\$\roott=\texttt{True}$\\$\DROP{}$}] (Z) at (2,2) {\Large$\Init$};
 \node [draw, rectangle, text width=2.8em, text height=0.9em, label={[label distance=0.0em]270:$\textit{Clean()}$}] (A) at (4,0) {\Large$\Roam$};
 \node [draw, rectangle, text width=1.9em, text height=0.9em, label={below:\TAKE{}}] (B) at (6,-2) {\Large$\Rr{1}$};
  \node [draw, rectangle, text width=1.9em, text height=0.9em, label={[label position=1em, align=center]270:$\textit{Clean()}$\\$\textit{Progress()}$}] (C) at (2,-2) {\Large$\Rr{0}$};
	 \node [draw, rectangle, text width=2.8em, text height=0.9em, label={[label position=1em, align=center]270:$\DROP{}$\\$\textit{Progress()}$}] (D) at (10,-2) {\Large$\Down$};
	 \node [draw, rectangle, text width=1.4em, text height=0.9em, label={[align=center]:$\DROP{}$\\$\textit{Progress()}$}] (E) at (8,2) {\Large$\Up$};
	 \node [draw, double, rectangle, text width=5.6em, text height=0.9em, label={End}] (F) at (12,2) {\Large$\Terminated$};
 \path 	(Z) edge [left, bend right] node { } (A)
				(A) edge [bend left] node [sloped,below] {$\Tok=1$} (B)
	      (B) edge [below, bend right] node {with $\Tok$} (D)
				(D) edge [right, bend right] node [align=left] {$\parent=\port$\\ $\TAKE{}$ then $\MOVE{}$} (E)
				(E) edge [loop below] node [align=left] {$\parent=\port$\\ $\TAKE{}$ then $\MOVE{}$} (E)
				(E) edge [bend right] node [sloped,above] {$\parent\neq\port$} (A)
				(A) edge [bend right] node [sloped, below] {$\Tok=0$} (C)
				(C) edge [loop left] node {$\Tok=0$} (C)
				(C) edge [below, bend right] node {$\Tok=1$} (B);
  \path	(D) edge [bend right] node [sloped,below] {$\parent\neq\port$} (A);
  \path	(E) edge [above] node[align=center] {$\port= 1$\\$\roott = \texttt{True}$} (F);
\end{tikzpicture}
\caption{Illustration of actions and state transitions in Algorithm~\ref{alg:token}.}
\label{fig:tokenflow}
\end{figure}
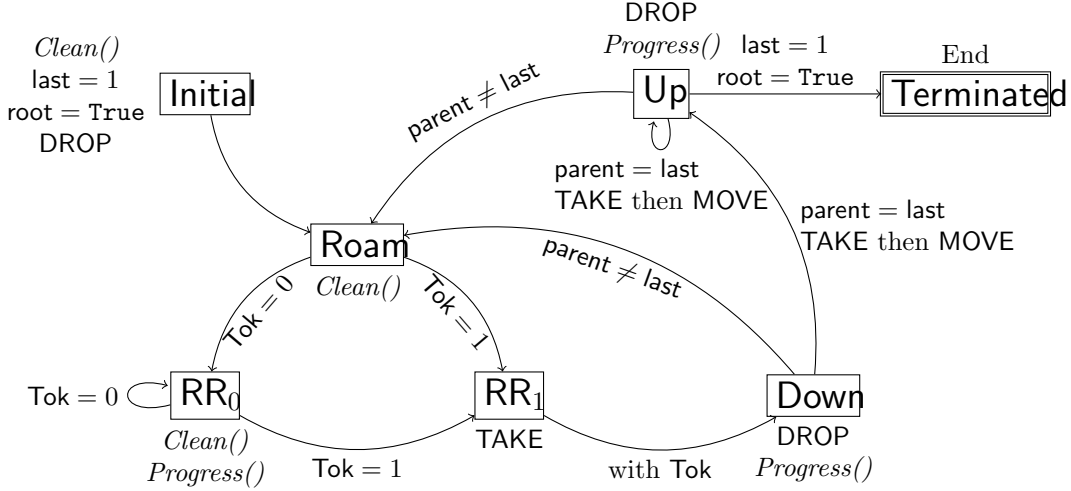%
\vspace*{-4mm}
\begin{theorem}
Algorithm \ref{alg:token} explores any tree and terminates at the starting node in at most $6(n-1)$ steps in \token model.
\end{theorem}
\vspace*{-1mm}
\begin{proof}
Using the induction from lemmas \ref{lem:sp0} and \ref{lem:spt}, we get that Algorithm \ref{alg:token} satisfies $\SP(t)$ for each step $t$.
First of all, from \ref{FirstClean}, all visited vertices are cleaned upon the first visits.\\
Consider some vertex $v$ and its arbitrary outport $p$. If $v\neq \Root$, then from \ref{PortOnceBefore} and \ref{PortOnceAfter}, $p$ may be used two times instantly after the incrementation of $v.\port(\cdot)$ (once when $v.\parent(\cdot)=\NULL$ and once when $v.\parent(\cdot)\neq\NULL$). Realize that, when $v=\Root$, then from \ref{AgUnderTok} and the fact, that any vertex with the token cannot be cleaned by $\Rr{0}$ and $\Roam$ moves, we claim that $\Root.\parent(\cdot)$ will always be $\NULL$. By \ref{PortOnceBefore}, $p$ may be used once instantly after the incrementation of $\Root.\port(\cdot)$. Moreover, regardless of the choice of $p$, it may be used also after $\Roam$ or $\Rr{1}$ action without a change of $v.\port(\cdot)$.\\ 
\textsf{Case 1.} Assume that $\Act(t)=\Rr{1}$ is taken via port $p$. From \ref{RootDir}$(t)$ and \ref{AgUnderTok}$(t)$, $p$ directs upwards towards the token.
Then $\Act(t+1)=\Down$ is performed downwards (from \ref{RootDir}$(t+1)$) via the same edge as $p$ (but with the opposite direction) and changes the position of the token to $\Tok(t+1)=v$ and establishes $v.\parent(t+1)\leftarrow p$. By \ref{RootDir}, $p$ cannot be used during $\Roam$ action. Realize that if $p=v.\port(s-1)$ for some moment $s>t$ and $\Act(s)=\Rr{1}$, then $\Tok(s-1)$ is one step towards $\Root$ from $v$ (from \ref{RootDir}$(s)$). Since the token can be moved only during $\Up$ and $\Down$ actions, then there exists a moment $s'$ such that $t<s'<s$ and $\Act(s')=\Up$ changes the position of the token from $v$ to $\Tok(s)$. However then $\Tok(s).\port(s) \neq \Tok(s).\port(s-1)$ ($\Tok(s)$ is not a leaf) and by \ref{PortOnceAfter}, \ref{OnePar} and \ref{ParOverTok}, it cannot be changed back since $\parent(\Tok(s))\neq \NULL$.\\
\textsf{Case 2.} Assume that $\Act(t)=\Roam$ is taken from node $v=\Tok(t-1)$ via port $p$ (downwards, by \ref{RootDir}$(t)$). 
The first action after time $t$, which leads to the token, is always $\Rr{1}$. However, as showed before, after $\Rr{1}$ action there is an instant $\Down$ action, which moves the token via $p$ to some vertex $w$ and sets $w.\parent \neq \NULL$. Hence by \ref{CleanOnce}, $\Roam$ action cannot be performed via $p$ once again. By \ref{RootDir}, $p$ cannot be used during $\Rr{1}$ action.

Our considerations show that each outport can be used at most $3$ times. Since each tree of size $n$ has $2(n-1)$ outports, Algorithm \ref{alg:token} terminates after at most $6(n-1)$ moves.

It remains to show that all vertices are then explored. Realize that Algorithm \ref{alg:token} terminates in $t+1$-st moment only when the agent is in $\Tok(t)$, $\Tok(t).\port(t)=1$, $\Tok(t).\roott = \texttt{True}$ and the state is $\Up$ at the end of $t$-th moment. From \ref{FirstClean}, we know that the first visit in some $v$ sets the $\roott$ flag to $\texttt{True}$ in $\Root$ and to $\texttt{False}$ in all other nodes, and this flag does not change ($\Roam$ and $\Rr{0}$ do not clean the vertices with the token, so by \ref{AgUnderTok}, $\Root$ cannot be cleaned by these moves). Hence the termination entails $\Tok(t)=\Root$. Since $\Root.\parent(t)$ $=\NULL$, then from \ref{PortOnceBefore} we know that each port $p$ in $\Root$ will be set by $\Root.\port(\cdot)$ only once (since the first moment) and \ref{RootProg} means that each subtree $T_{\Root,p}$ for $p\in\{1,\ldots,d_{\Root}\}$ were explored before returning to $\Root$, hence the exploration $T_{\Root}$ was completed.
%
%
\end{proof}
\newcommand{\As}{\mathrm{As}}
\newcommand{\Vs}{\mathrm{Vs}}
\newcommand{\Ps}{\mathrm{Ps}}
\vspace*{-1mm}
\section{Path with dirty memory}
\label{sec:line}
In this section we analyze exploration of paths with one bit of memory at each node and one bit at the agent. We show that in this setting, in the \dirtymem model, the exploration of the path sometimes requires $\Omega(n^2)$ steps.

\vspace*{2mm}
\customparagraph{Notation.}
In the following lower bound on the line, 
we will focus on the actions of the algorithm performed in vertices with degree $2$.
In such vertices, there are four possible inputs to the algorithm $S = \{(0,0),(0,1),(1,0),(1,1)\}$, where in a pair $(a,v)$, $a$ denotes a bit on the agent and $v$ is a bit saved on the vertex. 
Let us denote the sets of agent and vertex states by $\As$ and $\Vs$ respectively. We also use elements from $\Ps=\{0,1\}$ to denote ports in vertices of degree $2$. Morevoer, a shorthand notation $'$ indicates the inverted value of a bit. For example, if $a=0$, then $a'=1$ and vice versa. 
The goal of this section is to prove:
\vspace*{-1mm}
\begin{theorem}
\label{ref:thm}
Any algorithm that can explore any path in the \dirtymem model with one bit at the agent and one bit at each node requires time $\Omega(n^2)$ to explore some worst-case path with $n$ nodes.
\end{theorem}
\vspace*{-2mm}
Assume, for a contradiction, that there exists algorithm $\mathcal{A}$, which explores any path in $o(n^2)$ steps. Without a loss of generality, we may assume that the adversary always sets the agent in the middle point of the path.
For algorithm $\mathcal{A}$, for any $s \in S$, by $A(s), V(s), P(s)$ we denote respectively, the returned agent state, vertex state and the chosen outport in each vertex with degree $2$.
Moreover, we denote $R_3(s):=(A(s),V(s),P(s))$ and $R_2(s):=(A(s),V(s))$.
We will show that $\mathcal{A}$ either falls into an infinite loop or performs no faster than Rotor-Router.

A proof of \autoref{ref:thm} is divided naturally in several parts.
In the first one, we provide several properties of potential algorithms $\mathcal{A}$ that can eventually explore the path in $o(n^2)$ time.
These properties are utilized, among others, in the second part, where we prove that the agent does not change its internal state every time (Lemma~\ref{fact:1}), however it has to change the state of the vertex in each step (Lemma~\ref{lem:vchange}).
Further, we check two algorithms (called X and Y), which surprisingly turn out to explore the path even slower than the Rotor-Router algorithm (Lemma~\ref{lem:twoAlgs}). The analysis of those two algorithms turns out to be quite challenging.
Next we show that if three states of $\mathcal{A}$ return the same outport, then either it falls into an infinite loop or explores the path in $\Omega(n^2)$ steps.
In the latter part we consider the rest of amenable algorithms, which can either be reduced to previous counterexamples or fall into an infinite loop.
\begin{fact}
\label{wn1}
 $(\forall\; a\in \As)(\exists\; v\in \Vs)\; A(a,v)=a'$.
\end{fact}
\vspace*{-3mm}
\begin{proof}
Assume that $(\exists\; a \in \As)(\forall\; v\in \Vs)\; A(a,v)=a$.
If $a$ is the initial state of the agent, then the agent will never change its state before reaching an endpoint of the path.
Hence, the time to reach the endpoint cannot be faster than Rotor-Router algorithm, hence the adversary can initialize ports in such a way that the endpoint will be reached after $\Omega(n^2)$ steps.
If $a'$ is the starting state, then either the agent state never changes which reduces to the previous case or $A(a',w)=a$ for some $w \in \Vs$. Then, adversary sets $w$ as the initial state of the starting vertex and after the first step the state of the agent changes to $a$ and by the same argument as in the first case, the algorithm requires $\Omega(n^2)$ steps.
\end{proof}
\vspace*{-2mm}
\begin{fact}
\label{wn2}
$(\forall\; a \in \As)(\forall\;v \in \Vs)\;R_2(a,v) \neq (a,v)$.
\end{fact}
\vspace*{-2mm}
\begin{proof}
We will show that otherwise algorithm $\mathcal{A}$ falls into an infinite loop for some initial state of the path. Assume that $R_3(a,v) = (a,v,p)$ for some $a\in\As,\ v\in\Vs,\ p \in \Ps$. Then if $a$ is the starting agent state, the agent falls into an infinite loop on the left gadget from Fig.~\ref{fig:fact2} (double circle denotes the starting position of the agent and $a$ above the node indicates its initial state).
If $a'$ is the starting state, then by Fact~\ref{wn1}, $(\exists\; w\in \Vs)\;A(a',w)=a$. Let $q = P(a',w)$ and note that the agent falls into an infinite loop in the right gadget from Fig.~\ref{fig:fact2}.
\end{proof}
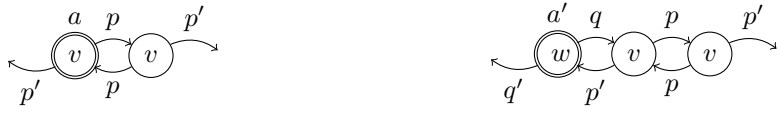
\begin{figure}[ht]
	\centering
	\begin{tikzpicture}[->]
		\node (X) at (0,0) { };
		\node [draw, double, circle, text width=.5em, text height=0.5em, label=$a$] (A) at (1,0) {$v$};
		\node [draw, circle, text width=.5em, text height=0.5em] (B) at (2,0) {$v$};
		\node (Y) at (3,0) { };
		\path 	(A) edge [above, bend left] node {$p$} (B)
		(B) edge [above, bend left] node {$p'$} (Y);
		\path	 (B) edge [below, bend left] node {$p$} (A)
		(A) edge [below, bend left] node {$p'$} (X);
	\end{tikzpicture}\hspace{3cm}
	\begin{tikzpicture}[->]
		\node (X) at (0,0) { };
		\node [draw, double, circle, text width=.5em, text height=0.5em, label=$a'$] (Z) at (1,0) {$w$};
		\node [draw, circle, text width=.5em, text height=0.5em] (A) at (2,0) {$v$};
		\node [draw, circle, text width=.5em, text height=0.5em] (B) at (3,0) {$v$};
		\node (Y) at (4,0) { };
		\path 	(Z) edge [above, bend left] node {$q$} (A)
		(A) edge [above, bend left] node {$p$} (B)
		(B) edge [above, bend left] node {$p'$} (Y);
		\path	 (B) edge [below, bend left] node {$p$} (A)
		(A) edge [below, bend left] node {$p'$} (Z)
		(Z) edge [below, bend left] node {$q'$} (X);
	\end{tikzpicture}
	\caption{Two looping gadgets for an algorithm with $R_2(a,v)=(a,v)$.}\label{fig:fact2}
\end{figure}
\vspace*{-3mm}

In the following series of facts we will prove some properties that $\mathcal{A}$ has to satisfy. 
\begin{lem}
\label{lem:vconst}
Fix any $a \in \As$ and $v\in \Vs$. Algorithm $\mathcal{A}$ does \textbf{not} satisfy the following condition:
\[
R_2(a,v)=(a',v) \land A(a',v) = A(a',v')~.
\]
\end{lem}
\begin{proof}
Assume that the formula holds for some $a \in \As$ and $v\in \Vs$ and we will show that in such a case $\mathcal{A}$ would fall into an infinite loop for some initializations of the path. 
From Fact \ref{wn1}, we obtain, that $A(a',v) = A(a',v') = a$. Denote $P(a,v)=p$ and $V(a',v)= y$.

If $y=v$, then let $P(a',v)=p'$, and the following initial state of the path results in an infinite loop: 
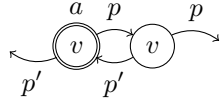
\begin{center}
\begin{tikzpicture}[->]
 \node (X) at (0,0) { };
 \node [draw, double, circle, text width=.5em, text height=0.5em, label=$a$] (A) at (1,0) {$v$};
 \node [draw, circle, text width=.5em, text height=0.5em] (B) at (2,0) {$v$};
  \node (Y) at (3,0) { };
 \path 	(A) edge [above, bend left] node {$p$} (B)
	  (B) edge [above, bend left] node {$p$} (Y);
  \path	 (B) edge [below, bend left] node {$p'$} (A)
	  (A) edge [below, bend left] node {$p'$} (X);
\end{tikzpicture}
\end{center}

In the opposite case, where $y = v'$, the agent would fall into an infinite loop on the following gadget:
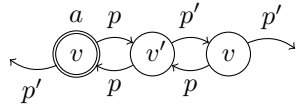
\begin{center}
\begin{tikzpicture}[->]
 \node (X) at (0,0) { };
 \node [draw, double, circle, text width=.5em, text height=0.5em, label=$a$] (Z) at (1,0) {$v$};
 \node [draw, circle, text width=.5em, text height=0.5em] (A) at (2,0) {$v'$};
 \node [draw, circle, text width=.5em, text height=0.5em] (B) at (3,0) {$v$};
  \node (Y) at (4,0) { };
 \path 	(Z) edge [above, bend left] node {$p$} (A)
		(A) edge [above, bend left] node {$p'$} (B)
	  (B) edge [above, bend left] node {$p'$} (Y);
  \path	 (B) edge [below, bend left] node {$p$} (A)
	  (A) edge [below, bend left] node {$p$} (Z)
		(Z) edge [below, bend left] node {$p'$} (X);
\end{tikzpicture}
\end{center}

Indeed, $R_3(a,v) = (a',v,p)$, hence in the first step, the agent moves to the middle vertex, changes its state to $a'$, and keeps $v$ as the state of the first vertex.
At the middle vertex when the agent state is $a'$, then it always changes its state to $a$ and moves to one of the outer vertices (independent on the vertex state of the middle node). But since the outer vertices are both in state $v$, the agent returns to the middle vertex, and so on, which is an infinite loop.
\end{proof}

The following fact characterises features of an algorithm, in which some state $a$ of the agent is flipped for both states of a vertex.
\begin{fact}
\label{wn3}
Fix any $a\in\As$. 
If $(\forall\; v\in \Vs)\; A(a,v)=a'$, then:
\begin{enumerate}[label=\roman*.,ref=\roman*.]
\item\label{fct:31} $(\forall\; w\in \Vs)\; R_2(a',w)\neq(a,w)$,
\item\label{fct:32} $(\forall\; w\in \Vs)\; V(a',w)=w'$,
\item\label{fct:33} $(\exists\; w\in \Vs)\; R_2(a',w)=(a,w')$.
\end{enumerate}
\end{fact}
\begin{proof}
For each $w\in\Vs$, Claim~\ref{fct:31} follows from Lemma~\ref{lem:vconst} for $a', w$.
Claim~\ref{fct:32} follows from Fact~\ref{wn2} and Claim~\ref{fct:31}.
Claim~\ref{fct:33} follows from Fact~\ref{wn1} and Claims~\ref{fct:31},~\ref{fct:32}.
\end{proof}
\begin{fact}
\label{fac:CASE4}
For any $v \in \Vs$, $R_2(a,v) \neq (a',v)$ or $R_2(a',v) \neq (a,v)$. 
\end{fact}
\begin{proof}
Let $P(a,v)=p$ and $P(a'v)=q$.
Observe that if $(\exists\; v\in\Vs)\; R_3(a,v) = (a',v,p)\ \land\ R_3(a',v) = (a,v,q)$, then the agent would fall into an infinite loop on the following gadget:
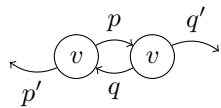
\begin{center}
\begin{tikzpicture}[->]
 \node (X) at (0,0) { };
 \node [draw, circle, text width=.5em, text height=0.5em] (A) at (1,0) {$v$};
 \node [draw, circle, text width=.5em, text height=0.5em] (B) at (2,0) {$v$};
  \node (Y) at (3,0) { };
 \path (A) edge [above, bend left] node {$p$} (B)
	  (B) edge [above, bend left] node {$q'$} (Y);
  \path  (B) edge [below, bend left] node {$q$} (A)
	   (A) edge [below, bend left] node {$p'$} (X);
\end{tikzpicture}
\end{center}
Indeed, if the agent's starting state is $a$ (respectively $a'$), the adversary chooses its starting position as the right (respectively left) vertex and the agent never exits from the gadget.
\end{proof}
\begin{lem}
\label{lem:ports}
Let $a\in \As$. If $(\forall\; v\in\Vs)\; P(a,v)=p$, then 
$
(\exists\; w\in \Vs)\; A(a',w)=a' \land P(a',w)=p'~.
$
\end{lem}
\begin{proof}
First observe that $(\exists\; w\in \Vs) P(a',w)=p'$, since otherwise the algorithm always takes the same port $p$ in every node with degree $2$ and thence can be trivially looped.

Observe that, if $(\forall\;w\in \Vs)\;P(a',w)=p' \Rightarrow A(a',w)=a$, then the agent always changes its state to $a$ whenever we takes port $p'$ and from the assumption, it always takes port $p$ in state $a$. Consider the following gadget (where $\ast$ denotes any vertex state, possibly different at each vertex):

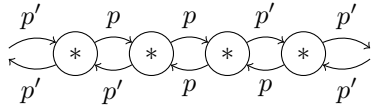
\begin{center}
\begin{tikzpicture}[->]
 \node (X) at (0,0) { };
 \node [draw, circle, text width=.5em, text height=0.5em] (A) at (1,0) {$*$};
 \node [draw, circle, text width=.5em, text height=0.5em] (B) at (2,0) {$\ast$};
 \node [draw, circle, text width=.5em, text height=0.5em] (C) at (3,0) {$\ast$};
  \node [draw, circle, text width=.5em, text height=0.5em] (D) at (4,0) {$\ast$};
  \node (Y) at (5,0) { };
 \path 	(X) edge [above, bend left] node {$p'$} (A)
		(A) edge [above, bend left] node {$p$} (B)
	        (B) edge [above, bend left] node {$p$} (C)
	  	(C) edge [above, bend left] node {$p'$} (D)
		(D) edge [above, bend left] node {$p'$} (Y);
  \path	 (Y) edge [below, bend left] node {$p'$} (D)
	         (D) edge [below, bend left] node {$p$} (C)
		(C) edge [below, bend left] node {$p$} (B)
		(B) edge [below, bend left] node {$p'$} (A)
		(A) edge [below, bend left] node {$p'$} (X);
\end{tikzpicture}
\end{center}
After entering to any of the outer vertices, the agent is certainly in state $a$, because it can only enter it via port $p'$.
Hence in the next step it takes port $p$. Thus it cannot exit from such a gadget. We obtain a contradiction, so the implication cannot be true, what completes the proof of the lemma.
\end{proof}
The following fact is a direct conclusion of Lemma~\ref{lem:ports}:
\begin{fact}
\label{wn5}
Let $a\in\As$. If $(\forall\; v\in\Vs)\; A(a',v)=a$, then $P(a,v)\neq P(a,v')$.
\end{fact}

The next lemma shows that the agent does not change its internal state in every step.

\begin{lem}
\label{fact:1}
Fix any $a\in \As$. If $(\forall\;v\in \Vs)\;A(a,v)=a'$, then $(\exists\; w\in \Vs)\;R_2(a',w)=(a,w')\ \land\ R_2(a',w')=(a',w)$. 
\end{lem}
\begin{proof}
By Fact \ref{wn3} we know that
\[
(\exists\; w\in \Vs)\;R_2(a',w)=(a,w')\ \land\ \left(R_2(a',w')=(a',w)\ \lor\ R_2(a',w')=(a,w)\right)~.
\]
In order to prove \autoref{fact:1}, let us assume that $R_2(a',w')=(a,w)$, what implies that the agent changes its internal state in each step (i.e. $(\forall\; a\in\As)(\forall\; v\in\Vs)\; A(a,v)=a'$).
A symmetric application of Fact \ref{wn3} shows that $(\exists\; v\in\Vs)\; R_2(a,v)=(a',v')\ \land\ R_2(a,v')=(a',v)$ (because $A(a,v')=a'$).
Hence we obtained that $(\forall\; a\in\As)(\forall\; v\in\Vs)\; R_2(a,v)=(a',v')$, what means that the states of vertices are changed upon each visit of the agent as well.
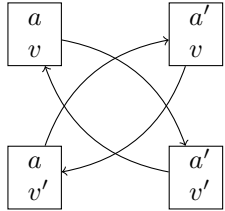
\begin{center}
\begin{tikzpicture}
  \matrix (m) [matrix of nodes, row sep=3em, column sep=4em, minimum width=2em, nodes={draw, align=left, text width=.5em, text height=0.5em}]
  {
    {$a$ \\ $v$}& {$a'$ \\ $v$} \\
     {$a$ \\ $v'$} & {$a'$ \\ $v'$} \\
     };
    \draw[->] (m-1-1) edge [above, bend left] (m-2-2);
    \draw[->] (m-1-2) edge [above, bend left] (m-2-1);
    \draw[->] (m-2-1) edge [above, bend left] (m-1-2);
    \draw[->] (m-2-2) edge [above, bend left] (m-1-1);
\end{tikzpicture}
\end{center}
We want to show that for such an algorithm:
\begin{equation}
  \label{eqn:differences}
   (\forall\; a\in\As)(\forall\; v\in\Vs)\; P(a,v)\neq P(a,v'). 
\end{equation}
Assume for contradition that for some $a \in \As,\ v \in \Vs,\ p \in \Ps$ we have $P(a,v) = P(a,v') = p$.
Observe that, in this case, in every step at a node with degree $2$ the agent changes its state, so at least one of every two consecutive steps at nodes with degree $2$ is via port $p$.
This shows, that on a path, where port $p$ always points towards the starting position of the agent, the agent cannot reach more than $2$ steps away from its starting position, hence it cannot explore the whole path. Thence Equation~\ref{eqn:differences} is proved for such the algorithm.

Now, let us consider the following two cases:

\noindent
\textbf{Case A: $P(a,v)=P(a',v)$}. Then, by Equation~\ref{eqn:differences}, $P(a,v')=P(a',v')$, so the outports chosen by the agent depend only on states of vertices (the agent does not use its internal state). 
Hence this algorithm is equivalent to Rotor-Router on vertices with degree $2$ and requires $\Omega(n^2)$ steps to explore the path, starting from a vertex in the middle of the path (at distance at least $\lfloor n/2 \rfloor$ from each endpoint).\\
\textbf{Case B: $P(a,v)\neq P(a',v)$}. Again from Equation~\ref{eqn:differences} we attain $P(a',v)=P(a,v')$ and $P(a,v)=P(a',v')$.
Let us think about an equivalence relation, where $W:=(a,v)\equiv(a',v')$ and $W':=(a,v')\equiv(a',v)$.
Note that the outports of $W$ and $W'$ are different. Moreover, since the agent changes its internal state at each step, then after an even number of steps on nodes with degree $2$, the agent is always in the same state.
Consider an agent starting in the middle of the path and analyze its walk until it reaches some of the endpoints. Between successive visits to any of the vertices $x$, the agent makes an even number of steps, hence the agent always visits $x$ in the same state, but each time the state of the vertex $x$ changes to opposite.
Thus upon consecutive visits, it chooses different outports. Hence the walk performed by the agent is equivalent to Rotor-Router and for some initialization of port labels it takes $\Omega(n^2)$ steps to reach any endpoint of the path.

We obtained a contradiction in both cases, which proves that $R_2(a',w')=(a',w)$.
\end{proof}
The following lemma shows that the agent changes the state of the vertex in each step at a node of degree $2$.
\begin{lem}
\label{lem:vchange}
$(\forall\; a\in\As)(\forall\; v\in\Vs)\; V(a,v)=v'$.
\end{lem}
\begin{proof}
Assume, that the algorithm does not satisfy the claim, so consider such $a\in\As,\ v\in\Vs$, that $V(a,v) = v$ and denote $P(a,v)=p$.
Then from Fact~\ref{wn2} we get $R_2(a,v) = (a',v)$. Moreover from Fact~\ref{fac:CASE4}, $R_2(a',v) \neq (a,v)$ and again from Fact~\ref{wn2}, $R_2(a',v) \neq (a',v)$, hence $V(a',v) = v'$.

We consider two cases:\\
\textbf{Case A: $R_2(a',v)=(a',v')$}. Observe that in this case from Lemma \ref{lem:vconst} we obtain $A(a',v)\neq A(a',v')=a$. Moreover, if $q = P(a',v) \neq P(a',v')=q'$, then the agent falls into infinite loop on the following gadget:

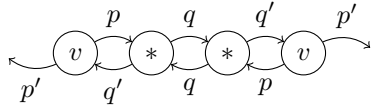
\begin{center}
\begin{tikzpicture}[->]
 \node (X) at (0,0) { };
 \node [draw, circle, text width=.5em, text height=0.5em, label=$ $] (A) at (1,0) {$v$};
 \node [draw, circle, text width=.5em, text height=0.5em, label=$ $] (B) at (2,0) {$*$};
 \node [draw, circle, text width=.5em, text height=0.5em, label=$ $] (C) at (3,0) {$*$};
  \node [draw, circle, text width=.5em, text height=0.5em, label=$ $] (D) at (4,0) {$v$};
  \node (Y) at (5,0) { };
 \path 	(A) edge [above, bend left] node {$p$} (B)
		(B) edge [above, bend left] node {$q$} (C)
	        (C) edge [above, bend left] node {$q'$} (D)
	        (D) edge [above, bend left] node {$p'$} (Y);
  \path	(D) edge [below, bend left] node {$p$} (C) 
  		(C) edge [below, bend left] node {$q$} (B)
	    (B) edge [below, bend left] node {$q'$} (A)
	    (A) edge [below, bend left] node {$p'$} (X);
\end{tikzpicture}
\end{center}
If the starting state of the agent is $a$, then the adversary chooses the starting position as the leftmost node and if it is $a'$, then of the middle nodes. Observe that if agent is in one of the outer vertices in state $a$, the agent changes its state to $a'$ and moves via port $p$ to one of the two vertices in the middle. Since $R_3(a',v) = (a',v',q)$ and $R_3(a',v') = (a,\ast,q')$ ($\ast$ takes any possible value), the next time the agent enters to the outer vertex (via port $q'$) will be in state $a$. 
And since $R_2(a,v) = (a',v)$ the states of both outer vertices remain $v$. This means that the agent cannot escape from such a gadget. 
 
Therefore $P(a',v)=P(a',v')$. From Lemma \ref{lem:ports} and Fact \ref{wn2} we obtain that $R_2(a,v')=(a,v)$ and $P(a',v)=P(a',v')\neq P(a,v')$.

If we had $R_2(a',v')=(a,v')$, then
by Lemma \ref{lem:vconst}, $A(a,v)\neq A(a,v')=a$, hence, from Fact \ref{wn2}, $R_2(a,v')=(a,v)$. By the symmetry, we conclude that $P(a,v)=P(a,v')\neq P(a',v)=P(a',v')$, hence ports depend only on the agent's state, so we obtain the following algorithm, with simple counterexamples: 
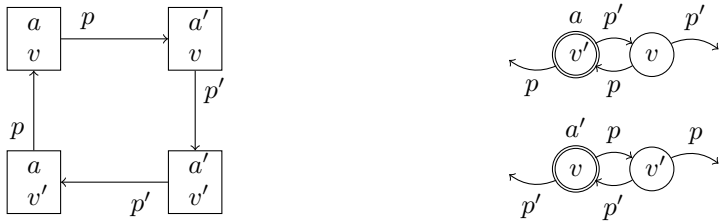
\begin{center}
\begin{minipage}{.4\textwidth}
\begin{tikzpicture}
  \matrix (m) [matrix of nodes, row sep=3em, column sep=4em, minimum width=2em, nodes={draw, align=left, text width=.5em, text height=0.5em}]
  {
    {$a$ \\ $v$}& {$a'$ \\ $v$} \\
     {$a$ \\ $v'$} & {$a'$ \\ $v'$} \\
     };
    \draw[->] (m-2-1) edge node [left, near start] {$p$} (m-1-1);
    \draw[->] (m-1-1) edge node [above, near start] {$p$} (m-1-2);
    \draw[->] (m-1-2) edge node [right, near start] {$p'$} (m-2-2);
    \draw[->] (m-2-2) edge node [below, near start] {$p'$} (m-2-1);
\end{tikzpicture}
\end{minipage}
\begin{minipage}{.4\textwidth}
\begin{tikzpicture}[->]
 \node (X) at (0,0) { };
 \node [draw, double, circle, text width=.5em, text height=0.5em, label=$a$] (A) at (1,0) {$v'$};
 \node [draw, circle, text width=.5em, text height=0.5em] (B) at (2,0) {$v$};
  \node (Y) at (3,0) { };
 \path (A) edge [above, bend left] node {$p'$} (B)
	  (B) edge [below, bend left] node {$p$} (A)
	  (A) edge [below, bend left] node {$p$} (X)
	  (B) edge [above, bend left] node {$p'$} (Y);
\end{tikzpicture}
\begin{tikzpicture}[->]
 \node (X) at (0,0) { };
 \node [draw, double, circle, text width=.5em, text height=0.5em, label=$a'$] (A) at (1,0) {$v$};
 \node [draw, circle, text width=.5em, text height=0.5em] (B) at (2,0) {$v'$};
  \node (Y) at (3,0) { };
 \path (A) edge [above, bend left] node {$p$} (B)
	  (B) edge [below, bend left] node {$p'$} (A)
	  (A) edge [below, bend left] node {$p'$} (X)
	  (B) edge [above, bend left] node {$p$} (Y);
\end{tikzpicture}
\end{minipage}
\end{center}

On the other hand, if $R_2(a',v')\neq (a,v')$, then by Fact~\ref{wn2} we obtain $V(a',v')=v$. Recall that $A(a',v')=a$, so $R_2(a',v')=(a,v)$.
Moreover, from Lemma~\ref{lem:ports}, we get that $A(a,v')=a$, because $A(a,v)=a'$. From Fact~\ref{wn2} we attain $R_2(a,v')=(a,v)$.

Now, if additionally $P(a,v')=p$, then from $P(a',v)=P(a',v')\neq P(a,v')$ 
we obtain $R_3(a',v')=(a,v,p')$ and $R_3(a',v)=(a',v',p')$, so we end up with the following algorithm and we present two gadgets for it:
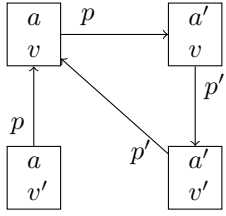
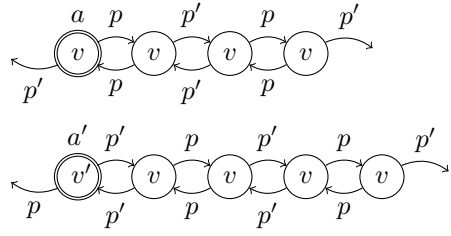
\begin{center}
\begin{minipage}{.4\textwidth}
\begin{tikzpicture}
  \matrix (m) [matrix of nodes, row sep=3em, column sep=4em, minimum width=2em, nodes={draw, align=left, text width=.5em, text height=0.5em}]
  {
    {$a$ \\ $v$}& {$a'$ \\ $v$} \\
     {$a$ \\ $v'$} & {$a'$ \\ $v'$} \\
     };
    \draw[->] (m-2-1) edge node [left, near start] {$p$} (m-1-1);
    \draw[->] (m-1-1) edge node [above, near start] {$p$} (m-1-2);
    \draw[->] (m-1-2) edge node [right, near start] {$p'$} (m-2-2);
    \draw[->] (m-2-2) edge node [below, near start] {$p'$} (m-1-1);
\end{tikzpicture}
\end{minipage}
\begin{minipage}{.4\textwidth}
\begin{tikzpicture}[->]
 \node (X) at (0,0) { };
 \node [draw, double, circle, text width=.5em, text height=0.5em, label=$a$] (A) at (1,0) {$v$};
 \node [draw, circle, text width=.5em, text height=0.5em] (B) at (2,0) {$v$};
  \node [draw, circle, text width=.5em, text height=0.5em] (C) at (3,0) {$v$};
  \node [draw, circle, text width=.5em, text height=0.5em] (D) at (4,0) {$v$};
  \node (Y) at (5,0) { };
 \path (A) edge [above, bend left] node {$p$} (B)
	  (B) edge [above, bend left] node {$p'$} (C)
	  (C) edge [above, bend left] node {$p$} (D)
	  (D) edge [above, bend left] node {$p'$} (Y);
  \path (D) edge [below, bend left] node {$p$} (C)
	  (C) edge [below, bend left] node {$p'$} (B)
	  (B) edge [below, bend left] node {$p$} (A)
	  (A) edge [below, bend left] node {$p'$} (X);
\end{tikzpicture}
\begin{tikzpicture}[->]
 \node (X) at (0,0) { };
 \node [draw, double, circle, text width=.5em, text height=0.5em, label=$a'$] (AA) at (1,0) {$v'$};
 \node [draw, circle, text width=.5em, text height=0.5em] (A) at (2,0) {$v$};
 \node [draw, circle, text width=.5em, text height=0.5em] (B) at (3,0) {$v$};
  \node [draw, circle, text width=.5em, text height=0.5em] (C) at (4,0) {$v$};
  \node [draw, circle, text width=.5em, text height=0.5em] (D) at (5,0) {$v$};
  \node (Y) at (6,0) { };
 \path (AA) edge [above, bend left] node {$p'$} (A) 
 	  (A) edge [above, bend left] node {$p$} (B)
	  (B) edge [above, bend left] node {$p'$} (C)
	  (C) edge [above, bend left] node {$p$} (D)
	  (D) edge [above, bend left] node {$p'$} (Y);
  \path (D) edge [below, bend left] node {$p$} (C)
	  (C) edge [below, bend left] node {$p'$} (B)
	  (B) edge [below, bend left] node {$p$} (A)
	  (A) edge [below, bend left] node {$p'$} (AA)
	  (AA) edge [below, bend left] node {$p$} (X);
\end{tikzpicture}
\end{minipage}
\end{center}

On the other hand, if alternatively $P(a,v')=p'$, 
then we similarly obtain $R_3(a',v')=(a,v,p)$ and $R_3(a',v)=(a',v',p)$ and the provided algorithm falls into infinite loop on the following gadgets:

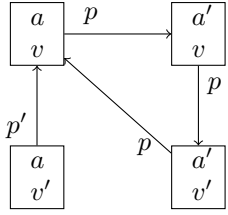
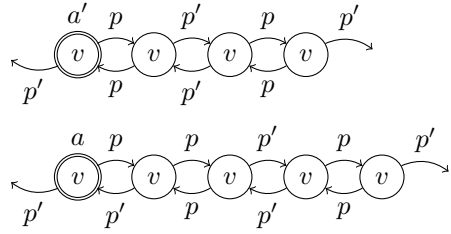
\begin{center}
\begin{minipage}{.4\textwidth}
\begin{tikzpicture}
  \matrix (m) [matrix of nodes, row sep=3em, column sep=4em, minimum width=2em, nodes={draw, align=left, text width=.5em, text height=0.5em}]
  {
    {$a$ \\ $v$}& {$a'$ \\ $v$} \\
     {$a$ \\ $v'$} & {$a'$ \\ $v'$} \\
     };
    \draw[->] (m-2-1) edge node [left, near start] {$p'$} (m-1-1);
    \draw[->] (m-1-1) edge node [above, near start] {$p$} (m-1-2);
    \draw[->] (m-1-2) edge node [right, near start] {$p$} (m-2-2);
    \draw[->] (m-2-2) edge node [below, near start] {$p$} (m-1-1);
\end{tikzpicture}
\end{minipage}
\begin{minipage}{.4\textwidth}
\begin{tikzpicture}[->]
 \node (X) at (0,0) { };
 \node [draw, double, circle, text width=.5em, text height=0.5em, label=$a'$] (A) at (1,0) {$v$};
 \node [draw, circle, text width=.5em, text height=0.5em] (B) at (2,0) {$v$};
  \node [draw, circle, text width=.5em, text height=0.5em] (C) at (3,0) {$v$};
  \node [draw, circle, text width=.5em, text height=0.5em] (D) at (4,0) {$v$};
  \node (Y) at (5,0) { };
 \path (A) edge [above, bend left] node {$p$} (B)
	  (B) edge [above, bend left] node {$p'$} (C)
	  (C) edge [above, bend left] node {$p$} (D)
	  (D) edge [above, bend left] node {$p'$} (Y);
  \path (D) edge [below, bend left] node {$p$} (C)
	  (C) edge [below, bend left] node {$p'$} (B)
	  (B) edge [below, bend left] node {$p$} (A)
	  (A) edge [below, bend left] node {$p'$} (X);
\end{tikzpicture}
\begin{tikzpicture}[->]
 \node (X) at (0,0) { };
  \node [draw, double, circle, text width=.5em, text height=0.5em, label=$a$] (AA) at (1,0) {$v$};
 \node [draw, circle, text width=.5em, text height=0.5em] (A) at (2,0) {$v$};
 \node [draw, circle, text width=.5em, text height=0.5em] (B) at (3,0) {$v$};
  \node [draw, circle, text width=.5em, text height=0.5em] (C) at (4,0) {$v$};
  \node [draw, circle, text width=.5em, text height=0.5em] (D) at (5,0) {$v$};
  \node (Y) at (6,0) { };
 \path  (AA) edge [above, bend left] node {$p$} (A) 
 	 (A) edge [above, bend left] node {$p$} (B)
	  (B) edge [above, bend left] node {$p'$} (C)
	  (C) edge [above, bend left] node {$p$} (D)
	  (D) edge [above, bend left] node {$p'$} (Y);
  \path (D) edge [below, bend left] node {$p$} (C)
	  (C) edge [below, bend left] node {$p'$} (B)
	  (B) edge [below, bend left] node {$p$} (A)
	  (A) edge [below, bend left] node {$p'$} (AA)
	  (AA) edge [below, bend left] node {$p'$} (X) ;
\end{tikzpicture}
\end{minipage}
\end{center}
This completes the analysis of Case A.\\
\noindent
\textbf{Case B: $R_2(a',v)=(a,v')$}.
From Lemma \ref{lem:vconst}, $A(a',v)\neq A(a',v')=a'$, so $R_2(a',v')=(a',v)$, due to Fact~\ref{wn2}.
If we had $R_2(a,v')=(a',v')$, then by symmetry we could swap $v$ with $v'$ and obtain Case A. Hence $R_2(a,v') \neq (a',v')$. 
Again, Fact~\ref{wn2} bears $V(a,v')=v$. 
Suppose that $P(a',v)=P(a',v')=p'$. Then $P(a,v')=p$ and $A(a,v')=a$ from Lemma \ref{lem:ports}, so $R_3(a,v')=(a,v,p)$, what provides Algorithm Q from Section~\ref{sec:line}, which can be looped by the adversary.
Thence we assume that either $P(a',v)\neq p'$ or $P(a',v')\neq p'$. Moreover, let us denote $P(a',v')=q$.
Consider the following gadget (if the agent starts in state $a$ it is initialized on one of the outer vertices and otherwise on one of the inner vertices, what is indicated above the nodes):
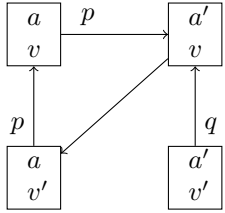
\begin{center}
\begin{minipage}{.4\textwidth}
\begin{tikzpicture}
  \matrix (m) [matrix of nodes, row sep=3em, column sep=4em, minimum width=2em, nodes={draw, align=left, text width=.5em, text height=0.5em}]
  {
    {$a$ \\ $v$}& {$a'$ \\ $v$} \\
     {$a$ \\ $v'$} & {$a'$ \\ $v'$} \\
     };
    \draw[->] (m-1-1) edge node [above, near start] {$p$} (m-1-2);
    \draw[->] (m-1-2) edge node [right, near start] { } (m-2-1);
    \draw[->] (m-2-1) edge node [left, near start] {$p$} (m-1-1);
    \draw[->] (m-2-2) edge node [right, near start] {$q$} (m-1-2);
\end{tikzpicture}
\end{minipage}
\begin{minipage}{.4\textwidth}
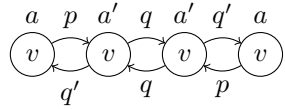
\begin{tikzpicture}[->]
 \node [draw, circle, text width=.5em, text height=0.5em, label=$a$] (A) at (1,0) {$v$};
 \node [draw, circle, text width=.5em, text height=0.5em, label=$a'$] (B) at (2,0) {$v$};
 \node [draw, circle, text width=.5em, text height=0.5em, label=$a'$] (C) at (3,0) {$v$};
  \node [draw, circle, text width=.5em, text height=0.5em, label=$a$] (D) at (4,0) {$v$};
 \path 	(A) edge [above, bend left] node {$p$} (B)
		(B) edge [above, bend left] node {$q$} (C)
	        (C) edge [above, bend left] node {$q'$} (D);
  \path	(D) edge [below, bend left] node {$p$} (C) 
  		(C) edge [below, bend left] node {$q$} (B)
	        (B) edge [below, bend left] node {$q'$} (A);
\end{tikzpicture}
\end{minipage}
\end{center}

Consider the following three invariants: 
\begin{itemize}
\item[I1] While the agent visits one of the inner vertices in state $a$, then the state of this vertex is always $v$.
\item[I2] The agent visits the outer vertices always in state $a$.
\item[I3] The state of the outer vertices is always $v$.
\end{itemize}
Notice, that as long as the invariants are not violated, the agent cannot escape the gadget because it always exits via port $q'$ from the outer vertices. Observe, that the invariants are initially true and assume that all three invariants were true up to some step $t$. We want to show that the invariants will remain true after this step.

First, we show that the agent cannot violate I3 in step $t+1$. In order to do so, it need to change the state of an outer vertex from $v$ to $v'$, so it has to visit this vertex in state $a'$ while entering to the vertex via port $q'$. Therefore, at the moment $t$, the agent should have internal state $a$, so since the invariant I1 was not violated, the appropriate vertex had state $v$. In consequence, $P(a,v)=q'$, hence $q=p'$. From the assumption, $P(a',v)=p$.
Note that the internal state of the agent cannot be changed to $a$ during the step $t-1$, from outer node to inner node, because it either violates I2 or I3.
Therefore such the move should be between two inner nodes, via port $q$.

Hence, from invariant I1, we may conclude that $P(a',v)=q=p'$, contrary to previous observation.
The same argument shows that the agent also cannot violate I2 in step $t+1$.

As long as invariants I2 and I3 are not violated, the agent always enters to inner vertices from outer vertices in state $a'$.
Thence, in order to violate I1, the state of the agent has to be changed to $a$ during the pass between the middle nodes (via port $q$), so $R_3(a',v)=(a,v',q)$ (otherwise the agent had to violate I1 before).
Then, from the assumption, $q$ cannot be $p'$, so from Lemma~\ref{lem:ports}, $P(a,v')=p'$. However, then the only way to go to outer node is via port $p'$, but then the agent violates I1.
Hence the agent cannot violate any of the invariants hence it cannot escape the gadget.
This completes the proof of Case B.
\end{proof}
\customparagraph{Algorithms $X$ and $Y$}
Consider the following two algorithms:

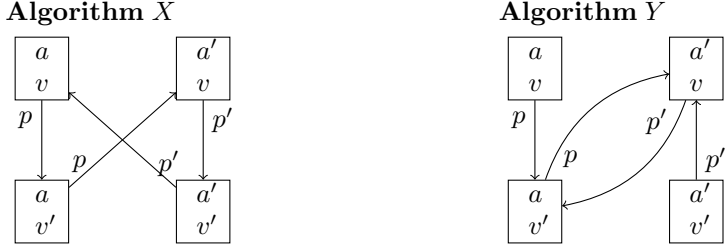
\begin{center}
\begin{minipage}{.4\textwidth}
\textbf{Algorithm\hspace*{1mm}$X$}
\linebreak
\begin{tikzpicture}
  \matrix (m) [matrix of nodes, row sep=3em, column sep=4em, minimum width=2em, nodes={draw, align=left, text width=.5em, text height=0.5em}]
  {
    {$a$ \\ $v$}& {$a'$ \\ $v$} \\
     {$a$ \\ $v'$} & {$a'$ \\ $v'$} \\
     };
    \draw[->] (m-1-1) edge node [left, near start] {$p$} (m-2-1);
    \draw[->] (m-1-2) edge node [right, near start] {$p'$} (m-2-2);
    \draw[->] (m-2-1) edge node [left, near start] {$p$} (m-1-2);
    \draw[->] (m-2-2) edge node [right, near start] {$p'$} (m-1-1);
\end{tikzpicture}
\end{minipage}
\begin{minipage}{.4\textwidth}
\textbf{Algorithm\hspace*{1mm}$Y$}
\linebreak
\begin{tikzpicture}
  \matrix (m) [matrix of nodes, row sep=3em, column sep=4em, minimum width=2em, nodes={draw, align=left, text width=.5em, text height=0.5em}]
  {
    {$a$ \\ $v$}& {$a'$ \\ $v$} \\
     {$a$ \\ $v'$} & {$a'$ \\ $v'$} \\
     };
    \draw[->] (m-1-1) edge node [left, near start] {$p$} (m-2-1);
    \draw[->,bend left] (m-1-2) edge node [left, very near start] {$p'$} (m-2-1);
    \draw[->,bend left] (m-2-1) edge node [right, very near start] {$p$} (m-1-2);
    \draw[->] (m-2-2) edge node [right, near start] {$p'$} (m-1-2);
\end{tikzpicture}
\end{minipage}
\end{center}

\begin{lem}
\label{lem:twoAlgs}
Algorithms $X$ and $Y$ require time $\Omega(n^2)$ to explore a path.
\end{lem}
\begin{proof}

Observe that in both algorithms X and Y, the state of the agent determines the outport chosen by the agent (because $P(a,\ast)=p$ and $P(a',\ast)=p'$). Moreover, the following property also holds for both algorithms:
\begin{equation}
\label{prop2}
(\forall\; v\in\Vs)(\forall\; a\in\As)(\exists\; b\in\As)\; R_2(a,v)=(b,v')\land R_2(a,v')=(b',v)~.
\end{equation}
Fix any algorithm $\mathcal{A}$ satisfying both aforementioned properties. Fix some arbitrary orientation of the line, where each port points either to the left or to the right. Endpoints of the line are its leftmost and rightmost nodes. The starting position of the agent is chosen as a node at distance at least $\lfloor n/2 \rfloor$ from both endpoints of the line.
We will show that for $\mathcal{A}$, for such initial position, for some initialization of port labels, the agent will perform a walk with the following pattern: two steps in some particullar direction interleaved with one step in the opposite direction. Moreover, after discovering a new node (i.e., visiting a node for the first time) the agent will reverse the direction of its walk. It is easy to see that such a walk needs $\Omega(n^2)$ steps to visit all the nodes.

We will show, that the adversary is able to force the agent to move according to a sequence that can be decomposed into blocks LR, RL, LRL and RLR where L denotes a move via port pointing to the left and R --- a move through the port pointing to the right. The adversary will define the port labels and initial states of nodes upon the first visits of the agent to these nodes.

In the first claim, we will show that when the agent visits a new node for the first time, it performs RL (or LR) block and reverses the direction of its walk. In this case we define the initialization of each newly discovered node upon the first visit of the agent at this node. It can be decoded into an explicit construction of a path initialization, different for algorithms X and Y.  
In the second claim we will show that when the agent visits a previously visited node, it performs RLR (or LRL) block and does not reverse the direction of its walk. 

\textit{Claim 1. When the agent visits a node $\beta$ for the first time with R-move (L-move) from node $\alpha$, the two consecutive steps (the step to node $\beta$ and the following step) is a RL (LR) block. The next block starts with L-move (R-move) from node $\alpha$.}

Assume, that the move to node $\beta$ is a R-move (the case of L-move is symmetric). 
Consider a situation that the agent uses R-move with a state $(b,w)$ and enters to $\beta$, with some agent state $c$. The vertex state of node $\alpha$ is then $w'$ (by (\ref{prop2})). The next move will be performed via port determined by $c$.
Let us denote $p=P(c,\ast)$). The adversary defines the left outport of the newly discovered node $\beta$ as $p$. Since the agent state determines the outport and the adversary wants the agent to perform L-move again, so the next time the agent will be in $\alpha$, it should have the opposite agent state, i.e. $b'$.
By (\ref{prop2}), there exists such $u$ that $R_2(c,u)=(b',u')$. Therefore, the adversary defines the initial state of $\beta$ as $u$. Hence the agent will perform a RL block. After the block, the agent is back in node $\alpha$ and its internal state is $b'$, hence the next step will be a L-move. This completes the proof of Claim 1.
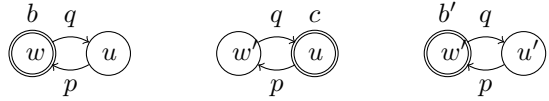
\begin{figure}
\begin{center}
\begin{tikzpicture}[->]
 \node [draw,double, circle, text width=.5em, text height=0.5em, label=$b$] (Z) at (1,0) {$w$};
 \node [draw, circle, text width=.5em, text height=0.5em, label=$ $] (A) at (2,0) {$u$};
 \path 	(Z) edge [above, bend left] node {$q$} (A);
  \path	(A) edge [below, bend left] node {$p$} (Z);
\end{tikzpicture}\hspace{1cm}
\begin{tikzpicture}[->]
 \node [draw, circle, text width=.5em, text height=0.5em, label=$ $] (Z) at (1,0) {$w'$};
 \node [draw, double, circle, text width=.5em, text height=0.5em, label=$c$] (A) at (2,0) {$u$};
 \path 	(Z) edge [above, bend left] node {$q$} (A);
  \path	(A) edge [below, bend left] node {$p$} (Z);
\end{tikzpicture}\hspace{1cm}
\begin{tikzpicture}[->]
 \node [draw, double, circle, text width=.5em, text height=0.5em, label=$b'$] (Z) at (1,0) {$w'$};
 \node [draw, circle, text width=.5em, text height=0.5em, label=$ $] (A) at (2,0) {$u'$};
 \path 	(Z) edge [above, bend left] node {$q$} (A);
  \path	(A) edge [below, bend left] node {$p$} (Z);
\end{tikzpicture}
\end{center}
\caption{Visualisation of consecutive steps of RL-block. Double circles denote the starting node and actual positions of the agent after each of the next two steps. Letters above the double circles indicates the actual internal states.}
\end{figure}


\textit{Claim 2. When the agent is in a node $\beta$ \textbf{not} for the first time after the first step of RLR-block (LRL-block), then the next block starts with R-move (L-move).}

We will investigate only RLR-block case (the second one is symmetric). 
Assume that at the beginning of step $t^{\ast}$, before RLR-block is preformed, we have the visit to this node):
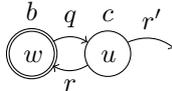
\begin{center}
\begin{tikzpicture}[->]
 \node [draw, double, circle, text width=.6em, text height=0.6em, label=$b$] (Z) at (1,0) {$w$};
 \node [draw, circle, text width=.5em, text height=0.6em, label=$c$] (A) at (2,0) {$u$};
 \node  (B) at (3,0) { };
 \path 	(Z) edge [above, bend left] node {$q$} (A)
		(A) edge [above, bend left] node {$r'$} (B);
  \path	(A) edge [below, bend left] node {$r$} (Z);
\end{tikzpicture}
\end{center}
Note that the agent state determines the outport, hence in order to perform RLR-block, the agent state has to be the same upon its both visits to the leftmost node.
During RLR-block, the state of the leftmost vertex is flipped twice (by (\ref{prop2})), hence its state after the execution of the block is identical to its state before the progress.
Note that the move from the rightmost node during the step $t^{\ast}+1$ changes the state of the rightmost vertex.
Realize that the last visit (before $t^{\ast}$) in the rightmost node was performed with the agent state $c$ and was followed by a L-move, so the state of the agent at the beginning of step $t^{\ast}+1$ is $c$.
We want the agent state to change to $c'$ during step $t^{\ast}+2$.
From \eqref{prop2}, there exists $v$ such that $A(b,v)=c'$, so $v=w'$. Therefore the situation after step $t^{\ast}+2$ is the following:
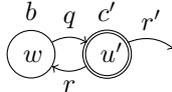
\begin{center}
\begin{tikzpicture}[->]
 \node [draw,circle, text width=.6em, text height=0.6em, label=$b$] (Z) at (1,0) {$w$};
 \node [draw, double,circle, text width=.5em, text height=0.6em, label=$c'$] (A) at (2,0) {$u'$};
 \node (B) at (3,0) { };
 \path 	(Z) edge [above, bend left] node {$q$} (A)
		(A) edge [above, bend left] node {$r'$} (B);
  \path (A) edge [below, bend left] node {$r$} (Z);
\end{tikzpicture}
\end{center}
Since the agent state determines the outport and $P(c,\ast)=r$, thus $P(c',\ast)=r'$, which means that the next block starts with R-move.

\textit{Claim 3.  When the agent visits a node $\beta$ \textbf{not} for the first time with R-move (L-move) from node $\alpha$ after the, the three consecutive steps (the step to node $\beta$ and two following steps) is a RLR (LRL) block.}


We will prove this claim by induction. Consider the earliest time step $t^{\ast}>4$, when the first part of the claim does not hold (observe that first two blocks are always LR and RL in some order, hence we do not lose generality). Assume that Claim 3 holds in all steps prior to $t^{\ast}$ and at the beginning of step $t^{\ast}$ we have the following situation:
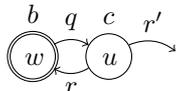
\begin{center}
\begin{tikzpicture}[->]
 \node [draw, double, circle, text width=.6em, text height=0.6em, label=$b$] (Z) at (1,0) {$w$};
 \node [draw, circle, text width=.5em, text height=0.6em, label=$c$] (A) at (2,0) {$u$};
 \node (Y) at (5,0) { };
 \path 	(Z) edge [above, bend left] node {$q$} (A)
		(A) edge [above, bend left] node {$r'$} (B);
  \path	(A) edge [below, bend left] node {$r$} (Z);
\end{tikzpicture}
\end{center}
Suppose that the agent wants to take two consecutive R-moves (the case of two consecutive L-moves is symmetric) during steps $t^{\ast}$ and $t^{\ast} + 1$.
By the assumption of this claim, the first R-move during step $t^{\ast}$ is to a node that has been visited before $t^{\ast}$.
Since R-move begins a new block and $t^{\ast}$ is the earliest time when this claim does not hold, then by Claim 2, the last move had to be R as well, so the previous block ended with sequence LR (either LR-block or RLR-block was used). 
Moreover after the last visit to the rightmost node, the agent moved to the leftmost node, hence $P(c,\ast)=r$ and $P(c',\ast)=r'$, thus since the agent is about to perform two consecutive R-moves, then $A(b,w)=c'$. Using the argument as in proofs of Claim 1, Claim 2, at the beginning of step $t^{\ast} - 2$ the situation looked as follows:
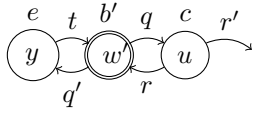
\begin{center}
\begin{tikzpicture}[->]
 \node [draw, circle, text width=.6em, text height=0.6em, label=$e$] (Z) at (1,0) {$y$};
 \node [draw, double, circle, text width=.5em, text height=0.6em, label=$b'$] (A) at (2,0) {$w'$};
 \node [draw, circle, text width=.6em, text height=0.6em, label=$c$] (B) at (3,0) {$u$};
  \node (C) at (4,0) { };
 \path 	(Z) edge [above, bend left] node {$t$} (A)
		(A) edge [above, bend left] node {$q$} (B)
	        (B) edge [above, bend left] node {$r'$} (C);
  \path (B) edge [below, bend left] node {$r$} (A)
	        (A) edge [below, bend left] node {$q'$} (Z);
\end{tikzpicture}
\end{center}

The last change of the state of the highlighted node was during step $t^{\ast}-2$. Note that at the end of step $t^{\ast}-2$, the leftmost node was visited. Moreover, the node $(c,u)$ was visited before this step (otherwise the agent cannot perform the double R-move in steps $t^{\ast}$ and $t^{\ast}+1$ due to Claim 1).
Therefore, let a step $s^{\ast}$ be the last moment, when the node $(b,w)$ was visited from the node $(c,u)$ before the step $t^{\ast}-2$. Note that thence during steps $s^{\ast}$ and $s^{\ast}+1$, the agent went left, so at $s^{\ast}$ the last move of a block ended by the sequence RL was performed.

If this block was LRL, then
%
during the step $s^{\ast}-2$, the agent has to go left, so since $P(c,\ast)=r$, the internal state is then $c$. However, from the previous argument about the changes of states during LRL-block, we can obtain a situation presented on the left of Figure~\ref{fig:sit} at the moment $s^{\ast}-2$.

Otherwise, the RL-block was performed during steps $s^{\ast}-1$ and $s^{\ast}$ and the situation before the step $s^{\ast}-1$ is given on the right of Figure~\ref{fig:sit}.
\begin{figure}[ht]
\begin{center}
\begin{tikzpicture}[->]
 \node [draw, circle, text width=.5em, text height=0.5em, label=$b$] (Z) at (1,0) {$w$};
 \node [draw, double, circle, text width=.5em, text height=0.5em, label=$c'$] (A) at (2,0) {$u'$};
 \path 	(Z) edge [above, bend left] node {$q$} (A);
  \path	(A) edge [below, bend left] node {$r$} (Z);
\end{tikzpicture}\hspace{3cm}
\begin{tikzpicture}[->]
 \node [draw, double, circle, text width=.5em, text height=0.5em, label=$b$] (Z) at (1,0) {$w$};
 \node [draw, circle, text width=.5em, text height=0.5em, label=$ $] (A) at (2,0) {$u'$};
 \path 	(Z) edge [above, bend left] node {$q$} (A);
  \path	(A) edge [below, bend left] node {$r$} (Z);
\end{tikzpicture}
\end{center}
\caption{A case of situation before LRL-block at the moment $s^{\ast}-2$ is presented on the left. A case before RL-block at the moment $s^{\ast}-1$ is given on the right}\label{fig:sit}
\end{figure}
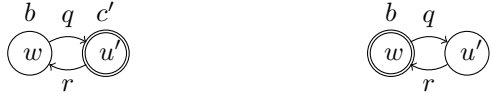
Therefore $A(b,w)=c$ (because $P(c,\ast)=r$), what is contrary to previously obtained $A(b,w)=c'$.
A symmetric situation holds if the agent wants to start a new block by LL.

Now we would like to estimate the number of moves needed to explore the graph by this algorithm.
Realize that after RL-block, LRL-block is repeated until the next new node is found by LR-block. Symmetrically with swapped R and L.
Let say that the line has $n$ nodes. Then the adversary may put the agent in the middle of the line in such the way that, so the agent has to performs $n-2$ RL-blocks or LR-blocks. After every RL-block (or LR-block), there is one more visited node, so the number of RLR-blocks (or LRL-blocks) that has to be executed before performing the next RL-block (or LR-block) increases by $1$. The last number of consecutive LRL-blocks or RLR-blocks is $n-2$ (because the last move is a single one), so the agent has to perform at least
\[
2(n-2)+1+\sum\limits_{i=1}^{n-2}3i=\frac{(n-2)(3n+1)}{2}+1=\Omega(n^2)
\]
moves in order to explore the graph. It worths to note that it is approximately $3$ times longer than Rotor-Router algorithm.
\end{proof}

%

\begin{lem}
\label{lem:threeports}
Assume that we have:
\begin{enumerate}
\item $(\forall\; a \in \As)(\forall\; v\in\Vs)\; V(a,v) = v'$,\label{v_swap}
\item function $P$ satisfies $P(s) = P(s') = P(s'')$ for three different state values $s, s', s''\in S$.
\end{enumerate}
Then an algorithm requires in the worst case $\Omega(n^2)$ steps to explore a path.
\end{lem}
\begin{proof}
Set $S$ has $4$ possible values, so if function $P$ returns the same port for all of them then the algorithm trivially cannot explore all paths of length $>2$.

Assume that for some $a \in \As$, $v\in \Vs$ and $p\in \Ps$, we have $P(a,v) = p$ and $P(s) = p'$ for all $s \neq (a,v)$.
Let us construct a path $\mathcal{P} = (v_1,v_2,\dots, v_n)$ of length $n$ such that a port from $v_i$ to $v_{i+1}$ is always $p$ (and port from $v_{i}$ to $v_{i-1}$ is $p'$). Let $v_1$ be the starting position of the agent and $v'$ be the initial state of each vertex.
Consider the exploration time of such the algorithm on $\mathcal{P}$.

To prove the time complexity, we show the following: \\
\textit{Claim: For any $i \in \{2,3,\dots n-1\}$, the agent visits node $v_1$ between first visits to $v_i$ and $v_{i+1}$.}\\
To prove the claim observe that, when entering $v_i$ for the first time, the last exit from each vertex $v_2,v_3,\dots,v_{i-1}$ was via port $p$. 
This means that the state of each of these nodes is $v'$ (it is because $V(a,v) = v'$ by (\ref{v_swap})).
Since the states of $v_i,v_{i-1},\dots,v_2$ are $v'$, and the agent is at $v_i$, it leaves each of these vertices via port $p'$ and walks straight to node $v_1$. This completes the proof of the claim.\\

The lemma is a direct consequence of the claim, because time to explore the whole path is at least $2 \cdot \sum_{i=1}^{n-2} i \in \Omega(n^2)$.
\end{proof}

We are ready to complete the proof of Theorem~\ref{ref:thm}
\begin{proof}[Proof of Theorem~\ref{ref:thm}]
From Lemma~\ref{fact:1} we know that the state of the agent is not changed during every step at the nodes with degree $2$. By Lemma~\ref{lem:vchange} the state of a vertex is flipped by the agent upon each visit to the vertex.
If the outport chosen by the agent is determined solely based on the state of the vertex then such algorithm would be equivalent to Rotor-Router (on vertices with degree $2$) hence would require time $\Omega(n^2)$. By Lemma~\ref{lem:threeports}, if function $P$ doe not take each of two ports in two different states $s\in S$, then the adversary may choose an arrangement in which the agent has to perform $\Omega(n^2)$ steps to explore the path. Consider the remaining possible algorithms.

\noindent
\textbf{Case A.}
Assume that the outport is defined by the state of the agent i.e.
\begin{equation}
\label{eqn:portstate}
(\exists\; p \in \Ps)(\forall\; v\in \Vs)\; P(a,v)=p \land P(a',v)=p'~.
\end{equation}
By Lemma~\ref{fact:1}, we know that there exists a state for which the agent does not change its internal state. Denote such state by $(a,v)$. 
Then by~\eqref{eqn:portstate}, $R_3(a,v)=(a,v',p)$ and because of Fact~\ref{fac:CASE4} and Lemma~\ref{lem:vchange}, we obtain $R_3(a,v')=(a',v,p)$.

If we had $A(a',v)=a'$, then similarly $R_3(a',v)=(a',v',p')$ and $R_3(a',v')=(a,v,p')$, what bears Algorithm X that requires $\Omega(n^2)$ steps by Lemma~\ref{lem:twoAlgs}.

Hence, suppose that $A(a',v)=a$, so $R_3(a',v)=(a,v',p')$ (by Fact~\ref{fac:CASE4} and Lemma~\ref{lem:vchange} once again).

If moreover $A(a',v')=a'$, then similar argument gives $R_3(a',v')=(a',v,p')$. Therefore we obtained Algorithm Y that also requires $\Omega(n^2)$ steps by Lemma~\ref{lem:twoAlgs}. 

Hence $A(a',v')=a$ and $R_3(a',v')=(a,v,p')$. The obtained algorithm satisfies (\ref{eqn:portstate}), 
so when the agent is walking on nodes with degree $2$, it cannot choose port $p'$ twice in a row, thus it cannot traverse from left to right the following gadget:
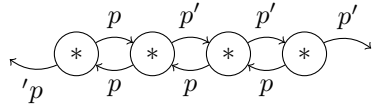
\begin{center}
\begin{tikzpicture}[->]
 \node (X) at (0,0) { };
 \node [draw, circle, text width=.5em, text height=0.5em, label=$ $] (A) at (1,0) {$\ast$};
 \node [draw, circle, text width=.5em, text height=0.5em, label=$ $] (B) at (2,0) {$\ast$};
 \node [draw, circle, text width=.5em, text height=0.5em, label=$ $] (C) at (3,0) {$\ast$};
  \node [draw, circle, text width=.5em, text height=0.5em, label=$ $] (D) at (4,0) {$\ast$};
  \node (Y) at (5,0) { };
 \path 	(A) edge [above, bend left] node {$p$} (B)
		(B) edge [above, bend left] node {$p'$} (C)
	        (C) edge [above, bend left] node {$p'$} (D)
	        (D) edge [above, bend left] node {$p'$} (Y);
  \path	(D) edge [below, bend left] node {$p$} (C) 
  		(C) edge [below, bend left] node {$p$} (B)
	    (B) edge [below, bend left] node {$p$} (A)
	    (A) edge [below, bend left] node {$'p$} (X);
\end{tikzpicture}
\end{center}
This completes the analysis of Case A.
\noindent
\textbf{Case B.} Now assume that 
\begin{equation}
\label{property}
(\forall\; p\in\Ps)(\forall\; a\in \As)(\exists\; v\in\Vs)\; P(a,v)=p~.
\end{equation}
Due to Fact \ref{wn1} we may assume without a loss of generality that $R_3(a,v)=(a',v',p)$.
Assume that $A(a',v)=a'$. Since the vertex state does not determine the outport Lemma~\ref{lem:vchange} and (\ref{property}), we get $R_3(a',v)=(a',v',p')$. Then also $R_3(a',v')=(a,v,p)$, because of Lemma~\ref{lem:vchange}, \eqref{property} and Fact~\ref{fac:CASE4}. Moreover, the agent state is not determined by the vertex state, so again from \eqref{property}, we get $R_3(a,v')=(a',v,p')$. The obtained algorithm (call it \textbf{Algorithm Z}) is incorrect due to the following counterexample:
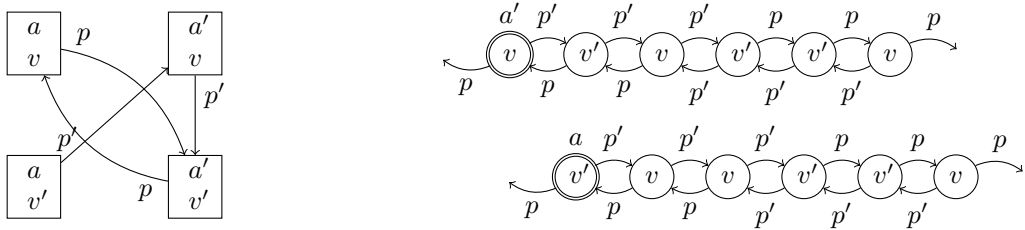
\begin{center}
\begin{minipage}{.4\textwidth}
\begin{tikzpicture}[->]
  \matrix (m) [matrix of nodes, row sep=3em, column sep=4em, minimum width=2em, nodes={draw, align=left, text width=.5em, text height=0.5em}]
  {
    {$a$ \\ $v$}& {$a'$ \\ $v$} \\
     {$a$ \\ $v'$} & {$a'$ \\ $v'$} \\
     };
    \draw[bend left] (m-1-1) edge node [above, very near start] {$p$} (m-2-2);
    \draw (m-1-2) edge node [right, near start] {$p'$} (m-2-2);
    \draw (m-2-1) edge node [left,near start] {$p'$} (m-1-2);
    \draw[bend left] (m-2-2) edge node [below, very near start] {$p$} (m-1-1);
\end{tikzpicture}
\end{minipage}
\begin{minipage}{.4\textwidth}
\hspace{-1cm}
\begin{tikzpicture}[->]
 \node (X) at (0,0) { };
 \node [draw, double, circle, text width=.5em, text height=0.5em, label=$a'$] (A) at (1,0) {$v$};
 \node [draw, circle, text width=.5em, text height=0.5em] (B) at (2,0) {$v'$};
  \node [draw, circle, text width=.5em, text height=0.5em] (C) at (3,0) {$v$};
  \node [draw, circle, text width=.5em, text height=0.5em] (D) at (4,0) {$v'$};
    \node [draw, circle, text width=.5em, text height=0.5em] (E) at (5,0) {$v'$};
      \node [draw, circle, text width=.5em, text height=0.5em] (F) at (6,0) {$v$};
  \node (Y) at (7,0) { };
 \path (A) edge [above, bend left] node {$p'$} (B)
	  (B) edge [above, bend left] node {$p'$} (C)
	  (C) edge [above, bend left] node {$p'$} (D)
	  (D) edge [above, bend left] node {$p$} (E)
	  (E) edge [above, bend left] node {$p$} (F)
	  (F) edge [above, bend left] node {$p$} (Y);
  \path (F) edge [below, bend left] node {$p'$} (E)
           (E) edge [below, bend left] node {$p'$} (D)
  	  (D) edge [below, bend left] node {$p'$} (C)
	  (C) edge [below, bend left] node {$p$} (B)
	  (B) edge [below, bend left] node {$p$} (A)
	  (A) edge [below, bend left] node {$p$} (X);
\end{tikzpicture}
\begin{tikzpicture}[->]
 \node (X) at (0,0) { };
 \node [draw, double, circle, text width=.5em, text height=0.5em, label=$a$] (A) at (1,0) {$v'$};
 \node [draw, circle, text width=.5em, text height=0.5em] (B) at (2,0) {$v$};
  \node [draw, circle, text width=.5em, text height=0.5em] (C) at (3,0) {$v$};
  \node [draw, circle, text width=.5em, text height=0.5em] (D) at (4,0) {$v'$};
    \node [draw, circle, text width=.5em, text height=0.5em] (E) at (5,0) {$v'$};
      \node [draw, circle, text width=.5em, text height=0.5em] (F) at (6,0) {$v$};
  \node (Y) at (7,0) { };
 \path (A) edge [above, bend left] node {$p'$} (B)
	  (B) edge [above, bend left] node {$p'$} (C)
	  (C) edge [above, bend left] node {$p'$} (D)
	  (D) edge [above, bend left] node {$p$} (E)
	  (E) edge [above, bend left] node {$p$} (F)
	  (F) edge [above, bend left] node {$p$} (Y);
  \path (F) edge [below, bend left] node {$p'$} (E)
           (E) edge [below, bend left] node {$p'$} (D)
  	  (D) edge [below, bend left] node {$p'$} (C)
	  (C) edge [below, bend left] node {$p$} (B)
	  (B) edge [below, bend left] node {$p$} (A)
	  (A) edge [below, bend left] node {$p$} (X);
\end{tikzpicture}
\end{minipage}
\end{center}
On the other hand if we have $A(a',v)=a$, then $R_3(a',v)=(a,v',p')$ from Lemma~\ref{lem:vchange} and \eqref{property}.

If moreover $A(a',v')=a$, then the same argument gives $R_3(a',v')=(a,v,p)$. However, if we swap $v$ with $v'$, then we obtain exactly Algorithm Z, which we know is incorrect. 

Hence from \eqref{property}, $R_3(a',v')=(a',v,p)$. Note that if we had $A(a,v') = a$, then $R_3(a,v') = (a,v,p')$ by similar argument, what bears algorithm would be equivalent to Algorithm X (it is sufficient to flip the vertex states and the ports).
Thence $A(a,v') = a'$, so together with Lemma~\ref{lem:vchange} and \eqref{property} we get $R_3(a,v')=(a',v,p')$ and we obtain the following \textbf{Algorithm R}:
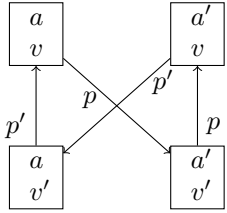
\begin{center}
\begin{minipage}{.4\textwidth}
\begin{tikzpicture}[->]
  \matrix (m) [matrix of nodes, row sep=3em, column sep=4em, minimum width=2em, nodes={draw, align=left, text width=.5em, text height=0.5em}]
  {
    {$a$ \\ $v$}& {$a'$ \\ $v$} \\
     {$a$ \\ $v'$} & {$a'$ \\ $v'$} \\
     };
    \draw  (m-2-1) edge node [left, near start] {$p'$} (m-1-1);
    \draw (m-2-2) edge node [right, near start] {$p$} (m-1-2);
    \draw (m-1-2) edge node [right, near start] {$p'$} (m-2-1);
    \draw (m-1-1) edge node [below, near start] {$p$} (m-2-2);
\end{tikzpicture}
\end{minipage}
\begin{minipage}{.4\textwidth}
\end{minipage}
\end{center}

Observe that Algorithm R has the following property. After taking port $p$, the agent is always in state $a'$ and after taking port $p$, the agent is always in state $a$. Using this property (and Lemma~\ref{lem:vchange}), we will show that Algorithm Z is equivalent to Rotor-Router. Indeed, let us organize port numbers on the path in the following pattern, where both incoming arcs to each node have the same port label:
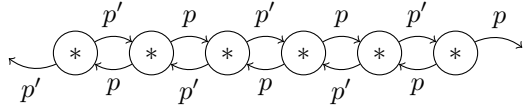
\begin{center}
\begin{tikzpicture}[->]
 \node (X) at (0,0) { };
 \node [draw, circle, text width=.5em, text height=0.5em, label=$ $] (A) at (1,0) {$\ast$};
 \node [draw, circle, text width=.5em, text height=0.5em, label=$ $] (B) at (2,0) {$\ast$};
 \node [draw, circle, text width=.5em, text height=0.5em, label=$ $] (C) at (3,0) {$\ast$};
  \node [draw, circle, text width=.5em, text height=0.5em, label=$ $] (D) at (4,0) {$\ast$};
  \node [draw, circle, text width=.5em, text height=0.5em, label=$ $] (E) at (5,0) {$\ast$};
  \node [draw, circle, text width=.5em, text height=0.5em, label=$ $] (F) at (6,0) {$\ast$};
  \node (Y) at (7,0) { };
 \path 	(A) edge [above, bend left] node {$p'$} (B)
		(B) edge [above, bend left] node {$p$} (C)
	        (C) edge [above, bend left] node {$p'$} (D)
	        (D) edge [above, bend left] node {$p$} (E)
	        (E) edge [above, bend left] node {$p'$} (F)
	        (F) edge [above, bend left] node {$p$} (Y);
  \path	(F) edge [below, bend left] node {$p$} (E) 
  		(E) edge [below, bend left] node {$p'$} (D) 
  		(D) edge [below, bend left] node {$p$} (C) 
  		(C) edge [below, bend left] node {$p'$} (B)
	    (B) edge [below, bend left] node {$p$} (A)
	    (A) edge [below, bend left] node {$p'$} (X);
\end{tikzpicture}
\end{center}
After such initialization of the port labels, we know that upon each visit to a vertex of degree $2$ the agent is in the same state. Note that this holds for all nodes of degree $2$, except the nodes that have one of the endpoints of the path as its neighbors.
But except four nodes (endpoints of the path and theirs neighbors), the algorithm behaves like Rotor-Router, which means that it chooses a different outport upon every visit to that node.
By choosing an appropriate initialization of vertex states, we can guarantee that the exploration time of Algorithm R is asymptotically the same as the worst-case exploration using Rotor-Router, which is $\Omega(n^2)$~\cite{BampasGHIKKR17}. This completes the analysis of Case B.

\noindent
\textbf{Case C.}
Assume that
\[
	(\exists\; p\in\Ps)(\exists\; a\in \As)(\forall\; v\in \Vs)\; P(a,v)=p~.
\]
Agent states do not determine the outports and vertex states always change, so Lemma~\ref{lem:ports} guarantees the assumptions of Lemma~\ref{lem:threeports} and hence in this case, the algorithm requires $\Omega(n^2)$ steps to explore a path.

By exhausting all possible cases, we obtain a contradiction with the assumption, that there exists an algorithm with time complexity $o(n^2)$, which completes the proof of Theorem~\ref{ref:thm}.
\end{proof}
\vspace*{-5mm}
\section{Conclusions and open problems}
\vspace*{-2mm}
One conclusion from our paper is that certain assumptions of the model of mobile agents can be exchanged. We showed, that in the context of linear time tree exploration, the assumption of clean memory at the nodes can be exchanged for a single token or the knowledge of the incoming port. The paper leaves a number of promising open directions. We showed that token and clean memory allow linear time exploration of trees, however, we have not ruled out the possibility that linear time exploration of trees is feasible without both these assumptions. Our lower bound suggests that memory $\omega(1)$ at the agent is probably necessary in \dirtymem model. Another open direction would be to consider different graph classes, or perhaps directed graphs. Finally, a very interesting future direction is to study dual-memory exploration with team of multiple mobile agents. Such approach could lead to even smaller exploration time, however, dividing the work between the agents in such models is very challenging since the graph is initially unknown.
\bibliographystyle{abbrv}
\bibliography{setup/biblio} 

\begin{thebibliography}{10}

\bibitem{AleliunasKLLR79}
R.~Aleliunas, R.~M. Karp, R.~J. Lipton, L.~Lov{\'{a}}sz, and C.~Rackoff.
\newblock Random walks, universal traversal sequences, and the complexity of
  maze problems.
\newblock In {\em 20th Annual Symposium on Foundations of Computer Science, San
  Juan, Puerto Rico, 29-31 October 1979}, pages 218--223. {IEEE} Computer
  Society, 1979.

\bibitem{BampasGHIKKR17}
E.~Bampas, L.~Gasieniec, N.~Hanusse, D.~Ilcinkas, R.~Klasing, A.~Kosowski, and
  T.~Radzik.
\newblock Robustness of the rotor-router mechanism.
\newblock {\em Algorithmica}, 78(3):869--895, 2017.

\bibitem{Bender:1994:PowerOfTeam}
M.~Bender and D.~Slonim.
\newblock {T}he {P}ower of {T}eam {E}xploration: {T}wo {R}obots {C}an {L}earn
  {U}nlabeled {D}irected {G}raphs.
\newblock In {\em Proceedings 35th Annual Symposium on Foundations of Computer
  Science}, 1994.

\bibitem{BenderFRSV02}
M.~A. Bender, A.~Fern{\'{a}}ndez, D.~Ron, A.~Sahai, and S.~P. Vadhan.
\newblock The power of a pebble: Exploring and mapping directed graphs.
\newblock {\em Inf. Comput.}, 176(1):1--21, 2002.

\bibitem{BerenbrinkCF15}
P.~Berenbrink, C.~Cooper, and T.~Friedetzky.
\newblock Random walks which prefer unvisited edges: Exploring high girth even
  degree expanders in linear time.
\newblock {\em Random Struct. Algorithms}, 46(1):36--54, 2015.

\bibitem{Cohen:2008:LabelGuidedFA}
R.~Cohen, P.~Fraigniaud, D.~Ilcinkas, A.~Korman, and D.~Peleg.
\newblock {L}abel-{G}uided {G}raph {E}xploration by a {F}inite {A}utomaton.
\newblock {\em ACM Trans. Algorithms}, 4(4):42:1--42:18, Aug. 2008.

\bibitem{dai1996improved}
H.~Dai and K.~E. Flannery.
\newblock Improved length lower bounds for reflecting sequences.
\newblock In {\em International Computing and Combinatorics Conference}, pages
  56--67. Springer, 1996.

\bibitem{DereniowskiDKPU15}
D.~Dereniowski, Y.~Disser, A.~Kosowski, D.~Pajak, and P.~Uznanski.
\newblock Fast collaborative graph exploration.
\newblock {\em Inf. Comput.}, 243:37--49, 2015.

\bibitem{DereniowskiKPU16}
D.~Dereniowski, A.~Kosowski, D.~Pajak, and P.~Uznanski.
\newblock Bounds on the cover time of parallel rotor walks.
\newblock {\em J. Comput. Syst. Sci.}, 82(5):802--816, 2016.

\bibitem{DiksFKP04}
K.~Diks, P.~Fraigniaud, E.~Kranakis, and A.~Pelc.
\newblock Tree exploration with little memory.
\newblock {\em J. Algorithms}, 51(1):38--63, 2004.

\bibitem{EfremenkoR09}
K.~Efremenko and O.~Reingold.
\newblock How well do random walks parallelize?
\newblock In {\em Approximation, Randomization, and Combinatorial Optimization.
  Algorithms and Techniques, 12th International Workshop, {APPROX} 2009, and
  13th International Workshop, {RANDOM} 2009, Berkeley, CA, USA, August 21-23,
  2009. Proceedings}, volume 5687 of {\em Lecture Notes in Computer Science},
  pages 476--489. Springer, 2009.

\bibitem{ElsasserS11}
R.~Els{\"{a}}sser and T.~Sauerwald.
\newblock Tight bounds for the cover time of multiple random walks.
\newblock {\em Theor. Comput. Sci.}, 412(24):2623--2641, 2011.

\bibitem{Feige95a}
U.~Feige.
\newblock A tight lower bound on the cover time for random walks on graphs.
\newblock {\em Random Struct. Algorithms}, 6(4):433--438, 1995.

\bibitem{Feige95}
U.~Feige.
\newblock A tight upper bound on the cover time for random walks on graphs.
\newblock {\em Random Struct. Algorithms}, 6(1):51--54, 1995.

\bibitem{FraigniaudGKP06}
P.~Fraigniaud, L.~Gasieniec, D.~R. Kowalski, and A.~Pelc.
\newblock Collective tree exploration.
\newblock {\em Networks}, 48(3):166--177, 2006.

\bibitem{FraigniaudI04}
P.~Fraigniaud and D.~Ilcinkas.
\newblock Digraphs exploration with little memory.
\newblock In {\em {STACS} 2004, 21st Annual Symposium on Theoretical Aspects of
  Computer Science, Montpellier, France, March 25-27, 2004, Proceedings},
  volume 2996 of {\em Lecture Notes in Computer Science}, pages 246--257.
  Springer, 2004.

\bibitem{FraigniaudIPPP05}
P.~Fraigniaud, D.~Ilcinkas, G.~Peer, A.~Pelc, and D.~Peleg.
\newblock Graph exploration by a finite automaton.
\newblock {\em Theor. Comput. Sci.}, 345(2-3):331--344, 2005.

\bibitem{FraigniaudIRT05}
P.~Fraigniaud, D.~Ilcinkas, S.~Rajsbaum, and S.~Tixeuil.
\newblock Space lower bounds for graph exploration via reduced automata.
\newblock In {\em Structural Information and Communication Complexity, 12th
  International Colloquium, {SIROCCO} 2005, Mont Saint-Michel, France, May
  24-26, 2005, Proceedings}, volume 3499 of {\em Lecture Notes in Computer
  Science}, pages 140--154. Springer, 2005.

\bibitem{GasieniecR08}
L.~Gasieniec and T.~Radzik.
\newblock Memory efficient anonymous graph exploration.
\newblock In {\em Graph-Theoretic Concepts in Computer Science, 34th
  International Workshop, {WG} 2008, Durham, UK, June 30 - July 2, 2008.
  Revised Papers}, volume 5344 of {\em Lecture Notes in Computer Science},
  pages 14--29, 2008.

\bibitem{KlasingKPS17}
R.~Klasing, A.~Kosowski, D.~Pajak, and T.~Sauerwald.
\newblock The multi-agent rotor-router on the ring: a deterministic alternative
  to parallel random walks.
\newblock {\em Distributed Comput.}, 30(2):127--148, 2017.

\bibitem{Kosowski13}
A.~Kosowski.
\newblock A \emph{{\~{O}}} (\emph{n\({}^{\mbox{2}}\)}) time-space trade-off for
  undirected \emph{s-t} connectivity.
\newblock In {\em Proceedings of the Twenty-Fourth Annual {ACM-SIAM} Symposium
  on Discrete Algorithms, {SODA} 2013, New Orleans, Louisiana, USA, January
  6-8, 2013}, pages 1873--1883. {SIAM}, 2013.

\bibitem{KosowskiP19}
A.~Kosowski and D.~Pajak.
\newblock Does adding more agents make a difference? {A} case study of cover
  time for the rotor-router.
\newblock {\em J. Comput. Syst. Sci.}, 106:80--93, 2019.

\bibitem{Koucky03}
M.~Kouck{\'{y}}.
\newblock Log-space constructible universal traversal sequences for cycles of
  length o(n\({}^{\mbox{4.03}}\)).
\newblock {\em Theor. Comput. Sci.}, 296(1):117--144, 2003.

\bibitem{lovasz1993random}
L.~Lov{\'a}sz.
\newblock Random walks on graphs.
\newblock {\em Combinatorics, Paul erdos is eighty}, 2(1-46):4, 1993.

\bibitem{MencPU17}
A.~Menc, D.~Pajak, and P.~Uznanski.
\newblock Time and space optimality of rotor-router graph exploration.
\newblock {\em Inf. Process. Lett.}, 127:17--20, 2017.

\bibitem{NonakaOSY10}
Y.~Nonaka, H.~Ono, K.~Sadakane, and M.~Yamashita.
\newblock The hitting and cover times of metropolis walks.
\newblock {\em Theor. Comput. Sci.}, 411(16-18):1889--1894, 2010.

\bibitem{OrtolfS14}
C.~Ortolf and C.~Schindelhauer.
\newblock A recursive approach to multi-robot exploration of trees.
\newblock In {\em Structural Information and Communication Complexity - 21st
  International Colloquium, {SIROCCO} 2014, Takayama, Japan, July 23-25, 2014.
  Proceedings}, volume 8576 of {\em Lecture Notes in Computer Science}, pages
  343--354. Springer, 2014.

\bibitem{Rollik:1980:AutomatenPlanaren}
H.~A. Rollik.
\newblock Automaten in planaren graphen.
\newblock {\em Acta Inf.}, 13(3):287–298, Mar. 1980.

\bibitem{sudo}
Y.~Sudo, F.~Ooshita, and S.~Kamei.
\newblock Self-stabilizing graph exploration by a single agent.
\newblock {\em CoRR}, abs/2010.08929, 2020.

\bibitem{YanovskiWB03}
V.~Yanovski, I.~A. Wagner, and A.~M. Bruckstein.
\newblock A distributed ant algorithm for efficiently patrolling a network.
\newblock {\em Algorithmica}, 37(3):165--186, 2003.

\end{thebibliography}
\end{document}